\documentclass[a4paper,UKenglish, cleveref]{lipics-v2019}
\hideLIPIcs
\nolinenumbers
\usepackage{microtype}
\usepackage{xspace}
\usepackage[normalem]{ulem}
\usepackage{color}
\usepackage{dsfont}
\usepackage{amsmath}
\usepackage{amssymb}
\usepackage{gensymb}
\usepackage{wrapfig}

\title{C-Planarity Testing of Embedded Clustered Graphs with Bounded Dual Carving-Width}

\titlerunning{C-Planarity Testing of Embedded C-Graphs with Bounded Dual Carving-Width}

\author{Giordano {Da Lozzo}}{Roma Tre University, Rome, Italy}{giordano.dalozzo@uniroma3.it}{}{}
\author{David Eppstein}{University of California, Irvine, USA}{eppstein@uci.edu}{}{}
\author{Michael T. Goodrich}{University of California, Irvine, USA}{goodrich@uci.edu}{}{}
\author{Siddharth Gupta}{Ben-Gurion University of the Negev, Beersheba, Israel}{siddhart@post.bgu.ac.il}{}{}

\authorrunning{G. Da Lozzo, D. Eppstein, M. T. Goodrich, and S. Gupta} 

\Copyright{Giordano Da Lozzo, David Eppstein, Michael T. Goodrich, and Siddharth Gupta}



\theoremstyle{plain}
\graphicspath{{figures/}}
\newcommand{\merge}{\raisebox{-2pt}{\includegraphics[page=1,scale=0.35]{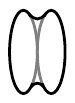}}\hspace{1pt}}

\Crefname{observation}{Observation}{Observations}
\Crefname{algorithm}{Algorithm}{Algorithms}
\Crefname{section}{Section}{Sections}
\Crefname{lemma}{Lemma}{Lemmas}
\Crefname{claim}{Claim}{Claims}
\Crefname{figure}{Fig.}{Figs.}
\Crefname{figure}{Fig.}{Figs.}
\Crefname{enumi}{Condition}{Conditions}

\newtheorem{observation}{Observation}

\renewcommand{\emph}[1]{\textcolor{lipicsblue}{\em #1}}

\def\etal{{\em et~al.}\xspace}

\definecolor{lipicsblue}{rgb}{0.08235294118,0.3098039216,0.537254902}
\definecolor{ourgreen}{rgb}{0,0.588,0.509}
\definecolor{ourred}{rgb}{1,0.3,0.3}

\hypersetup{breaklinks,colorlinks=true,urlcolor=lipicsblue,citecolor=lipicsblue,linkcolor=lipicsblue}

\newcommand{\cgraph}[2]{\ensuremath{\mathcal C^{#2}_{#1}(G^{#2}_{#1},\mathcal T^{#2}_{#1})}}

\newcommand{\NC}{\mathscr{NC}}
\newcommand{\RNC}{\mathscr{RNC}}
\newcommand{\CAT}{\mathscr{CAT}}

\newcommand{\calC}[1]{\ensuremath{{\cal C}^{#1}}\xspace}
\newcommand{\calT}[1]{\ensuremath{{\cal T}^{#1}}\xspace}

\DeclareMathOperator{\emw}{emw}
\DeclareMathOperator{\cw}{cw}
\DeclareMathOperator{\tw}{tw}
\DeclareMathOperator{\dual}{\delta}

\begin{document}
\maketitle

\begin{abstract}
For a \emph{clustered graph}, i.e, a graph whose vertex set is recursively partitioned into clusters, the {\sc C-Planarity Testing} problem asks whether it is possible to find a planar embedding of the graph and a representation of each cluster as a region homeomorphic to a closed disk such that {\bf 1.}~the subgraph induced by each cluster is drawn in the interior of the corresponding disk, {\bf 2.}~each edge intersects any disk at most once, and {\bf 3.}~the nesting between clusters is reflected by the representation, i.e., child clusters are properly contained in their parent cluster.
The computational complexity of this problem, whose study has been central to the theory of graph visualization since its introduction in 1995 \href{https://link.springer.com/chapter/10.1007%2F3-540-60313-1_145}{[Feng, Cohen, and Eades, {\em Planarity for clustered graphs}, ESA'95]}, has only been recently settled \href{https://arxiv.org/abs/1907.13086}{[Fulek and T{\'{o}}th, {\em Atomic Embeddability, Clustered Planarity, and Thickenability}, to appear at SODA'20]}. Before such a breakthrough, the complexity question was still unsolved even when the graph has a prescribed planar embedding, i.e, for \emph{embedded clustered graphs}.

We show that the {\sc C-Planarity Testing} problem admits a single-exponential single-parameter FPT algorithm for embedded clustered graphs, when parameterized by the carving-width of the dual graph of the input. This is the first FPT algorithm for this long-standing open problem with respect to a single notable graph-width parameter. Moreover, in the general case, the polynomial dependency of our FPT algorithm is smaller than the one of the algorithm by Fulek and T{\'{o}}th. To further strengthen the relevance of this result, we show that the {\sc C-Planarity Testing} problem retains its computational complexity when parameterized by several other graph-width parameters, which may potentially lead to faster algorithms.
\end{abstract}

\clearpage

\section{Introduction}
Many real-word data exhibit an intrinsic hierarchical structure that can be captured in the form of  clustered graphs, i.e., graphs equipped with a recursive clustering of their vertices. 
This graph model has proved very powerful to represent information at different levels of abstraction and drawings of clustered networks appear in a wide variety of application domains, such as software visualization, knowledge representation, visual statistics, and data mining. More formally, a \emph{clustered graph} (for short, \emph{c-graph}) is a pair $\cgraph{}{}$, where $G$ is the \emph{underlying graph} and $\mathcal T$ is the \emph{inclusion tree} of $\mathcal C$, i.e., a rooted tree whose leaves are the vertices of $G$. Each non-leaf node $\mu$ of $\mathcal T$ corresponds to a cluster containing the subset $V_\mu$ of the vertices of $G$ that are the leaves of the subtree of $\mathcal T$ rooted at $\mu$. Edges between vertices of the same cluster (resp.,\ of different clusters) \mbox{are \emph{intra-cluster edges} (resp.,\ \emph{inter-cluster edges}).}

A natural and well-established criterion for a readable visualization of a c-graph has been derived from the classical notion of graph planarity. 
\mbox{A~\emph{c-planar drawing}} of a c-graph $\cgraph{}{}$ (see \cref{fi:realizable_c}) is a planar drawing of $G$ together with a representation of each cluster~$\mu$ in~$\mathcal T$ as a region $D(\mu)$ homeomorphic to a closed disc such that:
(i) for each cluster $\mu$ in $\mathcal T$, region $D(\mu)$ contains the drawing of the subgraph $G[V_\mu]$ of $G$ induced by $V_\mu$;
(ii) for every two clusters $\mu$ and $\eta$ in $\mathcal T$, it holds $D(\eta) \subseteq D(\mu)$ if and only if $\eta$ is a descendant of $\mu$ in $\mathcal T$;
(iii) each edge crosses the boundary of any cluster disk at most once; and
(iv) the boundaries of no two cluster disks intersect.
A c-graph is \emph{c-planar} if it admits a c-planar drawing.

The {\sc C-Planarity Testing} problem,  introduced by Feng, Cohen, and Eades more than two decades ago~\cite{fce-pcg-95}, asks for the existence of a c-planar drawing of a c-graph.
Despite several algorithms having been presented in the literature to construct c-planar drawings of c-planar c-graphs with nice aesthetic features~\cite{DBLP:journals/dcg/AngeliniFK11,DBLP:conf/gd/BattistaDM01,DBLP:journals/jda/HongN10,DBLP:journals/dam/NagamochiK07}, determining the computational complexity of the {\sc C-Planarity Testing} problem has been one of the most challenging quests in the graph drawing research area~\cite{DBLP:conf/gd/BrandenburgEGKLM03,DBLP:conf/compgeom/CorteseB05,DBLP:journals/jgaa/Schaefer13}.
To shed light on the complexity of the problem, several researchers have tried to highlight its connections with other notoriously difficult problems in the area~\cite{DBLP:journals/cj/AngeliniL16,DBLP:journals/jgaa/Schaefer13}, as well as to consider relaxations~\cite{DBLP:journals/comgeo/AngeliniLBFPR15,DBLP:journals/jgaa/AngeliniLBFPR17,DBLP:journals/jgaa/AthenstadtC17,DBLP:journals/jgaa/DidimoGL08,DBLP:journals/jgaa/LozzoBFP18,DBLP:journals/corr/abs-1907-01630} and more constrained versions~\cite{DBLP:conf/soda/AkitayaFT18,DBLP:conf/isaac/AngeliniL16,addf-sprepg-17,ClusteredLevel15,DBLP:journals/dm/CorteseBPP09,DBLP:conf/sofsem/ForsterB04,DBLP:conf/compgeom/FulekK18}
of the classical notion of c-planarity. 
Algebraic approaches have also been considered~\cite{DBLP:journals/corr/abs-1305-4519,DBLP:conf/alenex/GutwengerMS14}. Only recently, Fulek and T{\'{o}}th settled the question by giving a polynomial-time algorithm for a generalization of the {\sc C-Planarity Testing} problem called Atomic Embeddability~\cite{FulekTothSODA19}.

A cluster $\mu$ is \emph{connected} if $G[V_\mu]$ is connected, and it is \emph{disconnected} otherwise. 
A c-graph is \emph{c-connected} if every cluster is connected. Efficient algorithms for the c-connected case have been known since the early stages of the research on the problem~\cite{cdfpp-cccg-06,DBLP:conf/latin/Dahlhaus98,fce-pcg-95}. Afterwards, polynomial-time algorithms have also been conceived for c-graphs satisfying other, weaker, connectivity requirements~\cite{CornelsenW06,Goodrich2006,Gutwenger2002}. 
A c-graph is \emph{flat} if each leaf-to-root path in $\calT{}$ consists of exactly two edges, that is, the clustering determines a partition of the vertex set; see, e.g.,~\cref{fi:realizable_a}. 
For flat c-graphs polynomial-time algorithms are known for several restricted cases~\cite{DBLP:conf/soda/AkitayaFT18,b-dppIIItcep-98,BlasiusR16,cdfk-atcpefcg-14,df-ectefgsf-13,FulekKMP15,DBLP:conf/compgeom/FulekK18,hn-sat2pepg-14,JelinekJKL08,Jelinkova2008}. 

\subparagraph{Motivations and contributions.}
In this paper, we consider the parameterized complexity of the \mbox{\sc C-Planarity Testing} problem for 
\emph{embedded c-graphs}, i.e., c-graphs with a prescribed combinatorial embedding; see also~\cite{BlasiusR16,ck-ssscp-12,degg-sfefcg-conf-17,Jelinkova2008} for previous work in this direction.

In \cref{se:preliminaries}, we show that the {\sc C-Planarity Testing} problem retains its  complexity when restricted to instances of bounded path-width and to connected instances of bounded tree-width. Such a result implies that the goal of devising an algorithm parameterized by graph-width parameters that are within a constant factor from tree-width (e.g., branch-width~\cite{DBLP:journals/jct/RobertsonS91}) or that are bounded by path-width (e.g., tree-width, rank-width~\cite{DBLP:journals/jgt/Oum08}, boolean-width~\cite{DBLP:conf/wg/AdlerBRRTV10}, and clique-width~\cite{DBLP:journals/siamcomp/CorneilR05}) and with a dependency on the input size which improves upon the one in~\cite{FulekTothSODA19} has to be regarded as a major algorithmic challenge.

Remarkably, before the results presented in~\cite{FulekTothSODA19}, the computational complexity of the problem was still unsolved even for instances with faces of bounded size, and polynomial-time algorithms were known only for ``small'' faces and in the flat scenario. Namely, Jelinkova~\etal~\cite{Jelinkova2008} presented a quadratic-time algorithm for $3$-connected flat c-graphs with faces of size at most $4$. Subsequently, Di Battista and Frati~\cite{df-ectefgsf-13} presented a linear-time/linear-space algorithm for embedded flat c-graphs with faces of size at most $5$. 

Motivated by the discussion above and by the results in \cref{se:preliminaries}, we focus our attention on embedded c-graphs~$\cgraph{}{}$ whose underlying graph $G$ has {\em both} bounded tree-width and bounded face size, i.e., instances such that the \emph{carving-width} of the dual $\dual(G)$ of $G$ is bounded.  
In \cref{se:algorithm}, we present an FPT algorithm based on a dynamic-programming approach on a bond-carving decomposition to solve the problem for embedded flat c-graphs, which is ultimately based on maintaining a succinct description of the internal cluster connectivity via non-crossing partitions. 
We remark that, to the best of the authors' knowledge, this is the first FPT algorithm for the {\sc C-Planarity Testing} problem, with respect to a {\em single} \mbox{graph-width parameter. More formally, we prove the following.}

\begin{theorem}\label{th:main}
{\sc C-Planarity Testing} can be solved in $O(2^{4 \omega + \log \omega} n + n^2)$ time for any $n$-vertex embedded flat c-graph $\cgraph{}{}$, where $\omega$ is the carving-width of $\dual(G)$, if a carving decomposition of $\dual(G)$ of width $\omega$ is provided, and in $O(2^{4 \omega + \log \omega} n + n^3)$ time, otherwise.
\end{theorem}

It is well know that the carving-width $\cw(\dual(H))$ of the dual graph $\dual(H)$ of a plane graph $H$ with maximum face size $\ell(H)$ and tree-width $\tw(H)$ satisfies the relationship $\cw(\dual(H)) \leq \ell(H)(\tw(H)+2)$~\cite{DBLP:journals/ijcga/BiedlV13,DBLP:journals/endm/BouchitteMT01}. Therefore, \cref{th:main} provides the first\footnote{The results in this paper were accepted for publication before the contribution of Fulek and T{\'{o}}th in~\cite{FulekTothSODA19}.} polynomial-time algorithm for instances of bounded face size and bounded tree-width, which answers an open question posed by Di Battista and Frati~\cite[Open~Problem~(ii)]{df-ectefgsf-13} for instances of bounded tree-width; also, since any $n$-vertex planar graph has tree-width in $O(\sqrt{n})$, it provides an $2^{O(\sqrt{n})}$ subexponential-time algorithm for instances of bounded face size, which improves the previous $2^{O(\sqrt{n}\log{n})}$ time bound presented in~\cite{degg-sfefcg-conf-17} for such instances. 

Further implications of \cref{th:main} for instances of bounded embedded-width and of bounded dual cut-width are discussed in \cref{se:related-params}. Moreover, we extend \cref{th:main} to get an FPT algorithm for general non-flat embedded c-graphs, whose running time is $O(4^{4 \omega + \log \omega} n + n^2)$ if a carving decomposition of $\dual(G)$ of width $\omega$ is provided, and is $O(4^{4 \omega + \log \omega} n + n^3)$, otherwise. 
The details for such an extension can be found in \cref{se:algorithm-non-flat}.

\section{Definitions and Preliminaries}\label{se:preliminaries} 

In this section, we give definitions and preliminaries that will be useful throughout.

\subparagraph{Graphs and connectivity.} 
A \emph{graph} $G=(V,E)$ is a pair, where $V$ is the set of \emph{vertices} of~$G$ and~$E$ is the set of \emph{edges} of $G$ ,i.e., pairs of vertices in $V$. A \emph{multigraph} is a generalization of a graph that allows the existence of multiple copies of the same edge.
The \emph{degree} of a vertex is the number of its incident edges.
We denote the \emph{maximum degree} of $G$ by $\Delta(G)$.
Also, for any $S \subseteq V$, we denote by $G[S]$ the subgraph of $G$ induced by the vertices in~$S$.

A graph is \emph{connected} if it contains a path between any two vertices. A \emph{cutvertex} is a vertex whose removal disconnects the graph. 
A connected graph containing no cutvertices is \emph{$2$-connected}.
The \emph{blocks} of a graph are its maximal $2$-connected subgraphs. In this paper, we only deal with connected graphs, unless stated otherwise.

\subparagraph{Planar graphs and embeddings.}
A drawing of a graph is \emph{planar} if it contains no edge crossings.
A graph is \emph{planar} if it admits a planar drawing.
Two planar drawings of the same graph are \emph{equivalent} if they determine the same \emph{rotation} at each vertex, i.e, the same circular orderings for the edges around each vertex.
A \emph{combinatorial embedding} (for short, \emph{embedding}) is an equivalence class of planar drawings.
A planar drawing partitions the plane into topologically connected regions, called \emph{faces}. 
The bounded faces are the \emph{inner faces}, while
the unbounded face is the \emph{outer face}.
A combinatorial embedding together with a choice for the outer face defines a \emph{planar embedding}.
An \emph{embedded graph} (resp.\ \emph{plane graph})~$G$ is a planar graph with a fixed combinatorial embedding (resp.\ fixed planar embedding). The \emph{length} of a face $f$ of $G$ is the number of occurrences of the edges of $G$ encountered in a traversal of the boundary of $f$. The \emph{maximum face size} $\ell(G)$ of $G$ is the maximum length over all~faces~of~$G$.

\subparagraph{C-Planarity.}
An \emph{embedded c-graph} $\calC{}(G,\calT{})$ is a c-graph whose underlying graph $G$ has a prescribed combinatorial embedding, and it is \emph{c-planar} if it admits a c-planar drawing that preserves the given embedding. Since we only deal with embedded c-graphs, in the remainder of the paper we will refer to them simply as c-graphs. Also, when $G$ and $\cal T$ are clear from the context, we simply denote $\calC{}(G,\calT{})$ as $\cal C$.
A \emph{candidate saturating edge} of~$\mathcal C$ is an edge not in $G$ between two vertices of the same cluster in $\mathcal T$ that are incident to the same face of $G$; refer to~\cref{fi:realizable_b}.
A c-graph $\cgraph{}{\prime}$ with $\mathcal T' = \mathcal T$ obtained by adding to~$\calC{}$ a subset~$E^+$ of its candidate saturating edges is a \emph{super c-graph} of~$\calC{}$; also, set $E^+$ is a \emph{planar saturation} if $G'$ is planar.
Further, c-graph $\mathcal C$ is \emph{hole-free} if there exists a face $f$ in $G$ such that when $f$ is chosen as the outer face for $G$ no cycle composed of vertices of the same cluster encloses a vertex of a different cluster in its interior. Finally, two c-graphs are \emph{equivalent} if and only if they are both c-planar or they are both not c-planar.

\begin{remark}
In this paper, we only consider c-graphs whose underlying graph is connected, unless stated otherwise.
\end{remark}

We will exploit the following characterization presented by Di~Battista and Frati~\cite{df-ectefgsf-13}, which holds true also for non-flat c-graphs although originally only proved for flat c-graphs.

\begin{theorem}[\cite{df-ectefgsf-13}, Theorem~1]\label{th:characterization}
A c-graph $\cgraph{}{}$ is c-planar if and only if:
\begin{enumerate}[(i)]
\item\label{con:1} $G$ is planar,
\item\label{con:2} $\mathcal C$ is hole-free, and
\item\label{con:3} 
there exists a super c-graph $\cgraph{}{*}$ of $\mathcal C$ such that $G^*$ is planar and $\mathcal C^*$ is c-connected.
\end{enumerate}
\end{theorem}

\cref{con:1} of~\cref{th:characterization} can be tested using any of the known linear-time planarity-testing algorithms. \cref{con:2} of~\cref{th:characterization} can be verified in linear time
as described by Di Battista and Frati \cite[Lemma~7]{df-ectefgsf-13}, by exploiting the linear-time algorithm for checking if an embedded, possibly non-flat, c-graph is hole-free presented by Dahlhaus~\cite{DBLP:conf/latin/Dahlhaus98}. Therefore, in the following we will assume that any c-graph satisfies these conditions and thus view the {\sc C-Planarity Testing} problem as one of testing~\cref{con:3}.

\subparagraph{Tree-width.}
A \emph{tree decomposition} of a graph $G$ is a tree $T$ whose nodes, called \emph{bags}, are labeled by subsets of vertices of $G$. For each vertex $v$ the bags containing $v$ must form a nonempty contiguous subtree of $T$, and for each edge $(u,v)$ of $G$ at least one bag of $T$ must contain both~$u$ and~$v$. The \emph{width} of the decomposition is one less than the maximum cardinality of any bag. The \emph{tree-width} $\tw(G)$ of $G$ is the minimum width of any of its tree decompositions.

\subparagraph{Cut-sets and duality.} Let $G = (V,E)$ be a connected graph and let $S$ be a subset of $V$. 
The partition $\{S,V\setminus S\}$ of $V$ is a \emph{cut} of $G$ and the set $(S,V\setminus S)$ of edges with an endpoint in $S$ and an endpoint in $V\setminus S$ is a \emph{cut-set} of $G$. Also, cut-set~$(S,V\setminus S)$ is a \emph{bond} if $G[S]$ and $G[V\setminus S]$ are both non-null and connected. 

For an embedded graph, the \emph{dual} $\dual(G)$ of $G$ is the planar multigraph that has a vertex~$v_f$, for each face $f$ of $G$, and an edge $(v_f,v_g)$, for each edge $e$ shared by faces $f$ and $g$. The edge $e$ is the \emph{dual edge} of $(v_f,v_g)$, and vice versa. 
Also, $\delta(G)$ is $2$-connected if and only if $G$ is $2$-connected.
\cref{fig:dual} shows a plane graph $G$ (black edges) and its dual $\dual(G)$ (purple edges); we will use these graphs as running examples throughout the paper.
The following duality is well known.

\begin{figure}
\begin{subfigure}{.19\textwidth}
\centering
\includegraphics[page=14, width=\textwidth]{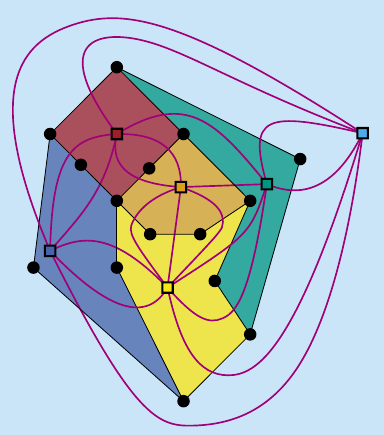}
\subcaption{}
\label{fig:dual}
\end{subfigure}
\begin{subfigure}{.40\textwidth}
\centering
\includegraphics[page=17, width=0.8\textwidth]{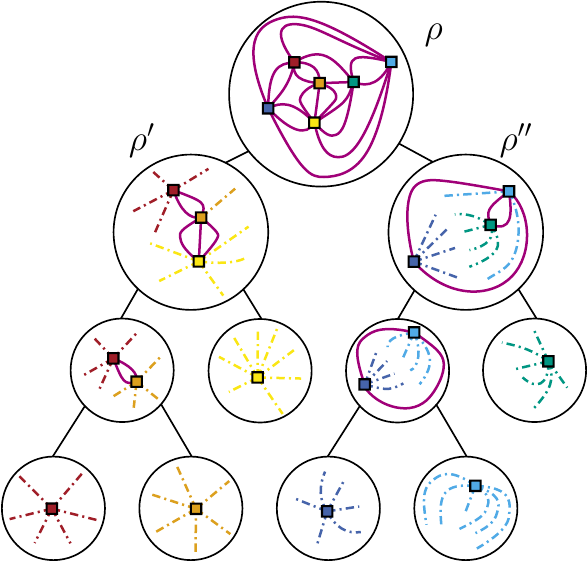}
\subcaption{}
\label{fi:decomposition_a}
\end{subfigure}
\begin{subfigure}{.40\textwidth}
\centering
\includegraphics[page=22, width=0.8\textwidth]{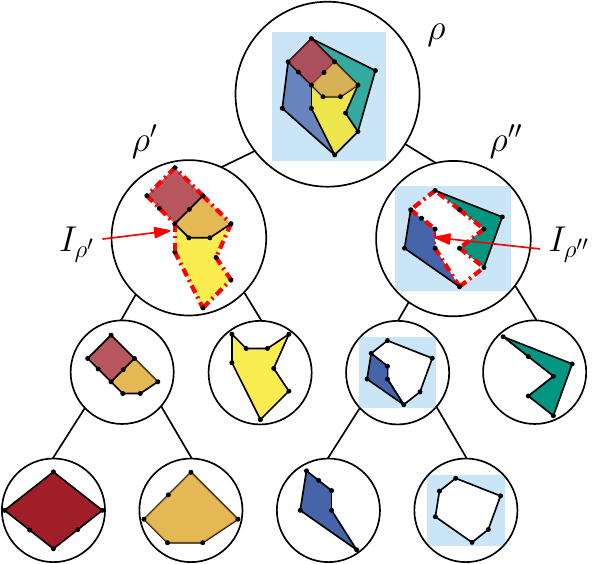}
\subcaption{}
\label{fi:decomposition_b}
\end{subfigure}
\caption{
(a) Running example: An embedded graph $G$ and its dual $\dual(G)$.
(b) A bond-carving decomposition $(D, \gamma)$ of the dual $\dual(G)$ of the graph $G$ in \cref{fig:dual}.
(c) The decomposition $(D, \gamma)$ where the vertices of $\dual(G)$ are replaced by the corresponding faces of~$G$.}
\label{fi:decomposition}
\end{figure}

\begin{lemma}[\cite{DBLP:books/daglib/0037866}, Theorem~14.3.1]\label{le:duality}
If $G$ is an embedded graph, then a set of edges is a cycle of $G$ if and only if their dual edges form a bond in $\dual(G)$.
\end{lemma}

\subparagraph{Carving-width.}
A \emph{carving decomposition} of a graph \mbox{$G = (V,E)$} is a pair $(D, \gamma)$, where $D$ is a rooted binary tree whose leaves are the vertices of $G$, and $\gamma$ is a function that maps the non-root nodes of $D$, called \emph{bags}, to cut-sets of $G$ as follows.
For any non-root bag $\nu$, let $D_\nu$ be the subtree of $D$ rooted at $\nu$ and let $\mathcal L_\nu$ be the set of leaves of $D_\nu$. Then, $\gamma(\nu) = (\mathcal L_\nu, V \setminus \mathcal L_\nu)$. The \emph{width} of a carving decomposition $(D,\gamma)$ is the maximum of $|\gamma(\nu)|$ over all bags $\nu$ in $D$. The \emph{carving-width} $\cw(G)$ of $G$ is the minimum width over all carving decompositions of~$G$.
The \emph{dual carving-width} is the carving-width of the dual of $G$.
A \emph{bond-carving decomposition} is a special kind of carving decomposition in which each cut-set is a bond of the graph; i.e., in a bond-carving decomposition every cut-set separates the graph into {two connected components}~\cite{DBLP:journals/talg/RueST14,DBLP:journals/combinatorica/SeymourT94}.

In this paper, we view a bond-carving decomposition of the vertices of the dual $\dual(G)$ of an embedded graph $G$ as a decomposition of the faces of $G$; see \cref{fi:decomposition_b}. A similar approach was followed in~\cite{DBLP:journals/ijcga/BiedlV13}.
In particular, due to the duality expressed by~\cref{le:duality},
the cut-sets~$\gamma(\nu)$ of the bags $\nu$ of $D$ correspond to cycles that can be used to {\em recursively partition} the faces of the primal graph, where these cycles are formed by the edges of the primal that are dual to those in each cut-set.

\subparagraph{Partitions.} 
Let $\mathcal Q = \{q_1, q_2, \dots, q_n\}$ be a ground set. A \emph{partition} of $\mathcal Q$ is a set $\{Q_1,\dots,Q_k\}$ of non-empty subsets $Q_i$'s of $\mathcal Q$, called \emph{parts}, such that $\mathcal Q= \bigcup^k_{i=1} Q_i$ and $Q_i \cap Q_j = \emptyset$, with $1 \leq i < j \leq k$. Observe that $k \leq |\mathcal Q|$.
Let now $\mathcal S = (s_1, s_2, \dots, s_n)$ be a \emph{cyclically-ordered set}, i.e., a set equipped with a circular ordering. Let $a$, $b$, and $c$ be three elements of $\mathcal S$ such that~$b$ appears after $a$ and before $c$ in the circular ordering of $\mathcal S$; we write $a \prec_b c$. A partition $P$ of $\mathcal S$ is \emph{crossing}, if there exist elements $a,c \in S_i$ and $b,d \in S_j$, with $S_i,S_j \in P$ and $i \neq j$, such that $a \prec_b c$ and $c \prec_d a$; and, it is \emph{non-crossing}, otherwise. We denote the set of all the non-crossing partitions of $\mathcal S$ by $\NC(\mathcal S)$. Note that, $|\NC(\mathcal S)|$ coincides with the \emph{Catalan number} $\mathscr{CAT}(n)$ of $n$, which satisfies  $\mathscr{CAT}(n) \leq 2^{2n}$.

\subsection{Relationship between Graph-Width Parameters and Connectivity}
In this section, we present reductions that shed light on the effect that the interplay between some notable graph-width parameters and the connectivity of the underlying graph have on the computational complexity of the {\sc C-Planarity Testing} problem. 

We will exploit recent results by Cortese and Patrignani, who proved the following:
\begin{enumerate}[(a)]
\item \label{prop:res1} Any $n$-vertex non-flat c-graph $\cgraph{}{}$ can be transformed into an equivalent $O(n \cdot h)$-vertex flat c-graph in quadratic time~\cite[Theorem~1]{DBLP:conf/gd/CorteseP18}, where $h$ is the height of $\cal T$.
\item \label{prop:res2} Any $n$-vertex flat c-graph can be turned into an equivalent $O(n)$-vertex \emph{independent flat c-graph}, i.e., a flat c-graph such that each non-root cluster induces an independent set, in linear time~\cite[Theorem~2]{DBLP:conf/gd/CorteseP18}.
\end{enumerate}
We remark that the reductions from~\cite{DBLP:conf/gd/CorteseP18} preserve the connectivity of the underlying graph.

\begin{theorem}\label{th:TreesOrForests}
Let $\cgraph{}{}$ be an $n$-vertex (flat) c-graph and let $h$ be the height of $\cal T$. In $O(n^2)$ time (in $O(n)$ time), it is possible to construct an $O(n\cdot h)$-vertex ($O(n)$-vertex) independent flat c-graph $\cgraph{}{\prime}$ that is equivalent to $\mathcal C$ such that: 
\begin{enumerate}[(i)]
\item \label{prop:pathwidth} $G'$ is a collection of stars or
\item \label{prop:treewidth} $G'$ is a tree.
\end{enumerate}
\end{theorem}

\begin{proof}
Let $\cgraph{}{}$ be an $n$-vertex (flat) c-graph. By the results~\ref{prop:res1} and~\ref{prop:res2} above, we can construct an $O(n\cdot h)$-vertex ($O(n)$-vertex) independent flat c-graph $\cgraph{}{+}$ equivalent to $\cal C$ in $O(n^2)$ time (in $O(n)$ time). Note that, $G^+$ only contains \mbox{inter-cluster edges.}

Let $e=(u,v)$ be an edge of $G^+$. Consider a c-graph $\cgraph{}{1}$ obtained from $\mathcal C^+$ as follows. First, subdivide the edge $e$ with two dummy vertices $u_e$ and $v_e$ to create edges $(u,u_e)$, $(u_e,v_e)$, and $(v_e,v)$. Then, delete the edge $(u_e,v_e)$. Finally, assign $u_e$ and $v_e$ to a new cluster $\mu_e$, and add~$\mu_e$ as a child of the root of the tree $\mathcal T^+$. 
To construct $\mathcal C'$ in case~(\ref{prop:pathwidth}), we perform the above transformation for all the edges of $G^+$. To construct $\mathcal C'$ in case~(\ref{prop:treewidth}), as long as the graph contains a cycle, we perform the above transformation for an edge $e$ of such a cycle. Since the construction of $C'$ from $C^+$ can be done in linear time both in case~(\ref{prop:pathwidth}) and~(\ref{prop:treewidth}), the running time follows. 

To show the correctness of the reduction, we prove that $\mathcal C^1$ is c-planar if and only if $\mathcal C^+$ is c-planar. 
Suppose first that $\mathcal C^+$ is c-planar and thus by~\cref{th:characterization}, it has a c-connected super c-graph $\cgraph{con}{+}$ with $G^+_{con}$ planar.  We now construct a c-connected super c-graph $\cgraph{con}{1}$ of $\mathcal C^1$ with $G^1_{con}$ planar as follows. First, initialize $G^1_{con} = G^+_{con}$ and $\mathcal T^1_{con} = \mathcal T^1$. As $G^+_{con}$ is a super graph of $G^+$, the edge $e$ also belongs to $G^+_{con}$ (and thus to $G^1_{con}$). Then, subdivide the edge $e$ with the  vertices $u_e$ and $v_e$ to create edges $(u,u_e)$, $(u_e,v_e)$, and $(v_e,v)$ in $G^1_{con}$. Clearly, $G^1_{con}$ is a super graph of $G^1$; also, since $e$ is an inter-cluster edge, all the clusters of $\mathcal T^+_{con}$ are still connected in $G^1_{con}$. Also, the cluster $\mu_e$ is connected in $G^1_{con}$ as the edge $(u_e,v_e)$ is present in $G^1_{con}$. Therefore, $\mathcal C^1_{con}$ is a c-connected super c-graph of $\mathcal C^1$ and thus, by~\cref{th:characterization}, $\mathcal C^1$ is c-planar.

Suppose now that $\mathcal C^1$ is c-planar and thus, by~\cref{th:characterization}, it has a c-connected super c-graph $\cgraph{con}{1}$ with $G^1_{con}$ planar. We now construct a c-connected super c-graph $\cgraph{con}{+}$ of $\mathcal C^+$ with $G^+_{con}$ planar as follows. First, initialize $G^+_{con} = G^1_{con}$ and $\mathcal T^+_{con} = \mathcal T^+$. As $\mathcal C^1_{con}$ is c-connected, the cluster $\mu_e$ is also connected, i.e., we have the edge $(u_e,v_e)$ in $G^1_{con}$. Consider the path $(u, u_e, v_e, v)$ and replace it with the edge $(u,v)$. Clearly, all the clusters of $\mathcal T^+_{con}$ are still connected in $G^+_{con}$, and $G^+_{con}$ is still planar. Therefore, $\mathcal C^+_{con}$ is a c-connected super c-graph of $\mathcal C^+$ and thus, by~\cref{th:characterization}, $\mathcal C^+$ is c-planar.
\end{proof}

We point out that by applying the reduction in the above proof without enforcing a specific embedding, \cref{th:TreesOrForests} also holds for general instances of the {\sc C-Planarity Testing} problem, i.e., non-embedded c-graphs. Moreover, since the reduction given in~\cite{DBLP:conf/gd/CorteseP18} also works for disconnected instances, applying the reduction of \cref{th:TreesOrForests} for case~(\ref{prop:pathwidth}) to a general disconnected instance $\cgraph{}{}$ would result in an equivalent independent flat c-graph $\cgraph{}{\prime}$ such that $G'$ is a collection of stars. An immediate, yet important, consequence of this discussion is that an algorithm with running time in $O(r(n))$ for flat instances whose underlying graph is a collection of stars would result in an algorithm with running time in $O(r(n))$ for flat instances and in $O(r(n^2) + n^2)$ for general, possibly non-flat, instances. 

The proof of the following lemma, which will turn useful in the following sections, is based on the duality expressed by~\cref{le:duality}.

\begin{figure}[t]
\centering
{\includegraphics[page=9,height=.28\columnwidth]{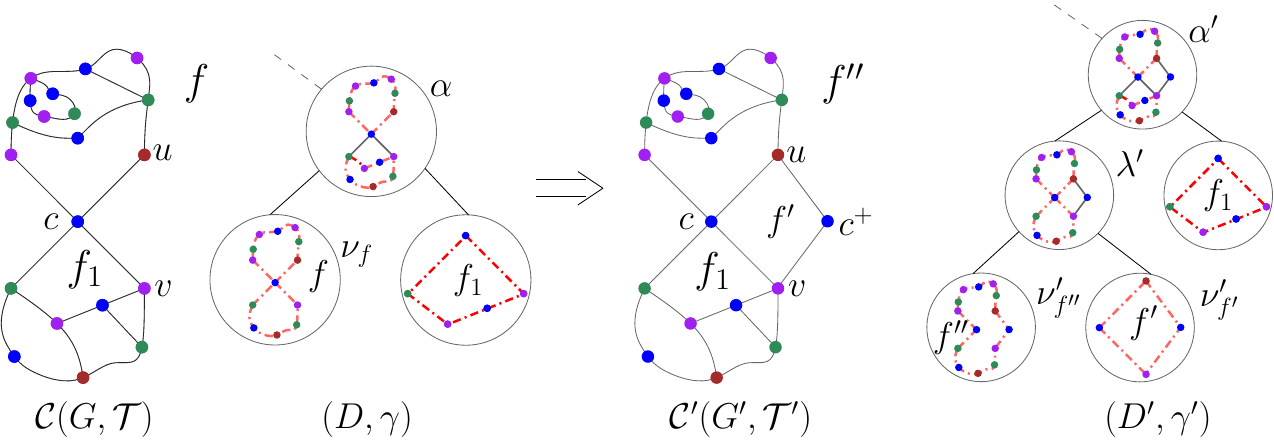}}
\caption{Reduction of~\cref{le:2connected} focused on cutvertex $c$. 
The transformation of $(D,\gamma)$ into $(D',\gamma')$ is shown. 
The red dashed edges are dual \mbox{to those in the cut-set of each bag.}
}\label{fi:2connected}
\end{figure}

\begin{lemma}\label{le:2connected}
Given an $n$-vertex c-graph $\cgraph{}{}$ and a carving decomposition $(D,\gamma)$ of $\dual(G)$ of width $\omega$, in $O(n)$ time, it is possible to construct an $O(n)$-vertex c-graph  $\cgraph{}{\prime}$ that is equivalent to $\mathcal C$ such that $G'$ is $2$-connected, and a carving decomposition $(D',\gamma')$ of $\dual(G')$ of width $\omega' = \max (\omega,4)$.
\end{lemma}

\begin{proof}
We construct $\cgraph{}{\prime}$ as follows. Let $\beta(G) < n$ be the number of blocks of $G$. 

Let $c$ be a cutvertex of $G$, let $\mu$ be the cluster that is the parent of $c$ in $\mathcal T$, and let~$(u,c)$ and~$(v,c)$ be two edges belonging to different blocks of $G$ that are incident to the same face~$f$ of $G$. Consider the c-graph $\cgraph{}{+}$ obtained from $\mathcal C$ by embedding a path $(u,c^+,v)$ inside $f$, where $c^+$ is a new vertex that we add as a child of $\mu$; see~\cref{fi:2connected}. We denote by $f'$ the face of~$G^+$ bounded by the cycle $(u,c^+,v,c)$ and by $f''$ the other face of $G^+$ incident~to~$c^+$. Clearly, this augmentation can be done in $O(1)$ time and $G^+$ contains $\beta(G)-1$ blocks. Also, $\mathcal C$ and $\mathcal C^+$ are equivalent. This is due to the fact that any saturating edge $(c,x)$ incident to $c$ and lying in $f$ (of a c-connected c-planar super c-graph of $\mathcal C$) can be replaced by two saturating edges $(x,c^+)$ lying in $f''$ and $(c^+,c)$ lying in $f'$ (of a c-connected c-planar super c-graph of $\mathcal C^+$), and vice versa.

We now show how to modify the carving decomposition $(D,\gamma)$ of $\dual(G)$ of width $\omega$ into a carving decomposition $(D^+,\gamma^+)$ of $\dual(G^+)$ of width $\omega^+ = \max(\omega,4)$ in $O(1)$ time.
Consider the leaf bag $\nu_f$ of $D$ corresponding to face~$f$ and let $\alpha$ be the parent of $\nu_f$ in $D$. 
We construct~$D^+$ from~$D$ as follows. First, we initialize $(D^+,\gamma^+) = (D, \gamma)$; in the following, we denote by~$\nu'$ the bag of $D^+$ corresponding to the bag~$\nu$ of~$D$. We remove $\nu'_f$ from~$D^+$, add a new non-leaf bag~$\lambda'$ as a child of~$\alpha'$ and two leaf bags $\nu'_{f'}$, corresponding to face~$f'$, and~$\nu'_{f''}$, corresponding to~face $f''$, as children of~$\lambda'$; refer to~\cref{fi:2connected}.
Further, we have $\gamma^+(\nu'_{f'})=\{(u,c),(u,c^+),(v,c),(v,c^+)\}$, 
$\gamma^+(\nu'_{f''})= \gamma(\nu'_f) \setminus \{(u,c),(v,c)\} \cup \{(u,c^+),(v,c^+)\}$,
$\gamma^+(\lambda') = \gamma(\nu'_f)$, and $\gamma^+(\nu') = \gamma(\nu)$, for any other bag $\nu$ belonging to both $D^+$ and $D$. In particular, the size of the edge-cuts defined by all the bags different from $\nu'_{f'}$ stays the same, while the size of the edge-cut of $\nu'_{f'}$ is $4$.
Therefore, $(D^+,\gamma^+)$ is a carving decomposition of~$\dual(G^+)$ of width~$\omega^+ = \max(\omega,4)$.

Repeating the above procedure, eventually yields a $2$-connected c-graph $\cgraph{}{\prime}$, with $|V(G')| = n + \beta(G)-1 = O(n)$, that is equivalent to $\mathcal C$ and a carving decomposition $(D',\gamma')$ of $\dual(G')$ of width $\omega' = \max(\omega,4)$. Since each execution of the above procedure takes $O(1)$ time and since the cutvertices and the blocks of $G$ can be computed in $O(n)$ time~\cite{DBLP:journals/cacm/HopcroftT73}, we have that $\mathcal C'$ and $(D',\gamma')$ can be constructed in $O(n)$ time. This concludes the proof. 
\end{proof}

\section{A Dynamic-Programming Algorithm for Flat Instances}\label{se:algorithm}

In this section, we present an FPT algorithm for the {\sc C-Planarity Testing} problem of flat c-graphs parameterized by the dual carving-width. We first describe a dynamic-programming algorithm to test whether a $2$-connected flat c-graph $\mathcal C$ is c-planar, by verifying whether $\mathcal C$ satisfies \cref{con:3} of~\cref{th:characterization}.
Then, by combining this result and~\cref{le:2connected}, we extend the algorithm to simply-connected instances. 

\subparagraph{Basic operations.}
Let $\cgraph{}{}$ be a flat c-graph. A partition $\{ S_1,\dots,S_k \}$ of $V' \subseteq V(G)$ is \emph{good} if, for each part $S_i$, there exists a non-root cluster $\mu$ such that all the vertices in $S_i$ belong to $\mu$; also, we say that the part $S_i$ \emph{belongs} to the cluster $\mu$. Further, a partition of a cyclically-ordered set ${\cal S} \subseteq V(G)$ is \emph{admissible} if it is both good and non-crossing. We define the binary operator $\uplus$, called \emph{generalized union}, that given two good partitions $P'$ and $P''$ of ground sets $\mathcal Q'$ and $\mathcal Q''$, respectively, returns a good partition $P^* = P' \uplus P''$ of $\mathcal Q' \cup \mathcal Q''$ obtained as follows. Initialize $P^* = P' \cup P''$. Then, as long as there exist $Q_i,Q_j \in P^*$ such that $Q_i \cap Q_j \neq \emptyset$, replace sets $Q_i$ and $Q_j$ with their union $Q_i \cup Q_j$ in~$P^*$. We have the following.

\begin{figure}[t]
\centering
\begin{subfigure}{.38\textwidth}
    \centering
  \includegraphics[page=25,height=0.48\textwidth]{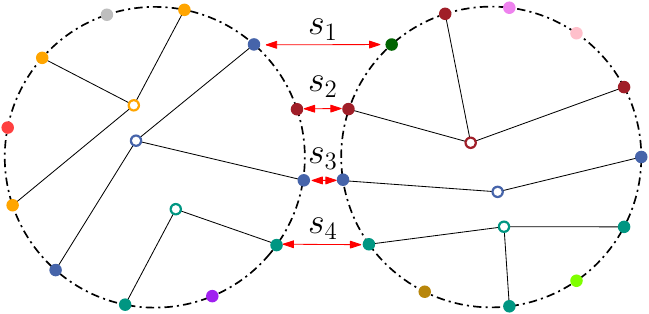}
  \subcaption{$G(P')$ and $G(P'')$}
  \label{fi:projection_a}
\end{subfigure}
\begin{subfigure}{.38\textwidth}
    \centering
  \includegraphics[page=26,height=0.48\textwidth]{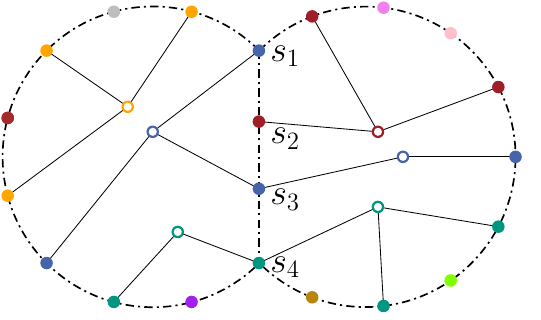}
  \subcaption{$H$}
  \label{fi:projection_b}
\end{subfigure}
\begin{subfigure}{.19\textwidth}
    \centering
  \includegraphics[page=27,height=0.96\textwidth]{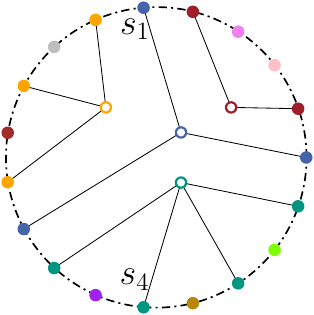}
  \subcaption{$G(P^*)$}
  \label{fi:projection_c}
\end{subfigure}
\caption{Illustrations for the definition of bubble merge.}\label{fi:projection}
\end{figure}

\begin{lemma}\label{le:linear-union}
$P^* = P' \uplus P''$ can be computed in $O(|\mathcal Q'|+|\mathcal Q''|)$ time.
\end{lemma}

\begin{proof}
We show how to compute $P^*$ in $O(|\mathcal Q'|+|\mathcal Q''|)$ time. Let $P' = \{Q'_1,\dots,Q'_{k'}\}$ and $P'' = \{Q''_1,\dots,Q''_{k''}\}$. We construct an undirected graph $G^*$ as follows. First, for each element $a \in \mathcal Q' \cup \mathcal Q''$, we add a vertex $v_a$ to $V(G^*)$. Then, for each part $Q'_i \in P'$ (resp.\ $Q''_i \in P''$), we add a vertex $v_{Q'_i}$ (resp.\ $v_{Q''_i}$) to $V(G^*)$. Finally, for each element $x \in Q'_i$ (resp.\ $x \in Q''_i$), we add an edge $(v_x, v_{Q'_i})$ (resp.\ $(v_x, v_{Q''_i})$) to $E(G^*)$. As $k' \leq |\mathcal Q'|$ and $k'' \leq |\mathcal Q''|$, $|V(G^*)| \leq |\mathcal Q'|+|\mathcal Q''|+k'+k'' \leq 2(|\mathcal Q'|+|\mathcal Q''|)$ and $E(G^*) \leq |\mathcal Q'|+|\mathcal Q''|$. 

Observe that, for every part $Q'_i \in P'$ (resp.\ $Q''_i \in P''$), the set $V_{Q'_i} = \{v_{Q'_i}\} \cup \{v_x | x \in Q'_i\}$ (resp.\ $V_{Q''_i} = \{v_{Q''_i}\} \cup \{v_x | x \in Q''_i\}$) induces a star in $G^*$. Therefore, there exist~$Q'_i \in P'$ and $Q''_j \in P''$ such that $Q'_i \cap Q''_j \neq \emptyset$ if and only if $V_{Q'_i \cup Q''_j} = \{v_{Q'_i}\} \cup \{v_{Q''_j}\} \cup \{v_x | x \in Q'_i \cup Q''_j\} $ induces a connected subgraph of $G^*$. We now perform a breadth-first search (BFS) in~$G^*$ to find all the connected components of $G^*$. Let $C= \{C_1, C_2, \dots, C_l\}$ be the set of all the connected components. For every $C_i \in C$, we construct a new part $Q^*_i$ by removing from $C_i$ the vertices corresponding to parts in $P'$ and $P''$. The set $\{Q^*_1, \dots, Q^*_l\}$ is our required $P^*$. Since the BFS runs in $O(|V(G^*)|+|E(G^*)|)$ time, the lemma follows.
\end{proof}
 
Let $P$ be a good partition of the ground set $\mathcal{Q}$ and let $\mathcal{Q}' \subset \mathcal{Q}$. The \emph{projection} of $P$ \emph{onto}~$\mathcal{Q}'$, denoted as $P|_{\mathcal{Q}'}$, is the good partition of $\mathcal{Q}'$ obtained from $P$ by first replacing each part $S_i \in P$ with $S_i \cap \mathcal{Q}'$ and then removing empty parts, if any. 

An admissible partition $P$ of a cyclically-ordered set $\mathcal S$ can be naturally associated with a $2$-connected plane graph $G(P)$ as follows. The outer face of $G(P)$ is a cycle $C(P)$ whose vertices are the elements in $\mathcal S$ and the clockwise order in which they appear along $C(P)$ is the same as in~$\mathcal S$. Also, for each part $S_i \in P$ such that $|S_i| \geq 2$, graph $G(P)$ contains a vertex $v_i$ in the interior of $C(P)$ that is adjacent to all the elements in $S_i$, i.e., removing all the edges of $C(P)$ yields a collection of stars, whose central vertices are the $v_i$'s, and isolated vertices. We say that $G(P)$ is the \emph{cycle-star} associated with $P$; see, e.g., \cref{fi:projection_c}.

We also extend the definitions of \emph{generalized union} and \emph{projection} to admissible partitions by regarding the corresponding cyclically-ordered sets as unordered.

Let $P'$ and $P''$ be two admissible partitions of cyclically-ordered sets $\mathcal S'$ and $\mathcal S''$, respectively, with the following properties (where $\mathcal S_\cap = \{s_1, s_2, \dots, s_k\}$ denotes the set of elements that are common to $\mathcal S'$ and $\mathcal S''$):
(i) $|\mathcal S_\cap| \geq 2$ and $\mathcal S' \cup \mathcal S'' \setminus \mathcal S_\cap \neq \emptyset$, 
(ii) the elements of~$\mathcal S_\cap$ appear consecutively both in $\mathcal S'$ and $\mathcal S''$, and 
(iii) the cyclic ordering of the elements in~$\mathcal S_\cap$ determined by~$\mathcal S'$ is the reverse of the cyclic order of these elements determined by~$\mathcal S''$.
We define the binary operator $\merge$, called \emph{bubble merge}, that returns an admissible partition $P^* = P' \merge P''$ obtained as follows. Consider the cycle-stars $G(P')$ and $G(P'')$ associated with $P'$ and $P''$, respectively. First, we identify the vertices corresponding to the same element of $\mathcal S_\cap$ in both $C(P')$ and $C(P'')$ (see~\cref{fi:projection_a}) to obtain a new plane graph $H$ (see \cref{fi:projection_b}). Observe that $H$ is $2$-connected since $G(P')$ and $G(P'')$ are $2$-connected and since $|\mathcal S_\cap| \geq 2$; therefore, the outer face~$f_H$ of~$H$ is a simple cycle. Second, we traverse $f_H$ clockwise to construct a cyclically-ordered set $\mathcal S^* \subseteq \mathcal S' \cup \mathcal S''$ on the vertices of $f_H$. Finally, we set $P^* = (P' \uplus P'')|_{\mathcal S^*}$. $P^*$ is good by the definition of generalized union. The fact that $P^*$ is a non-crossing partition of $\mathcal S^*$ follows immediately by the planarity of $H$ (see~\cref{fi:projection_c}). 
\cref{le:linear-union} and the fact that the graph $H$ can be easily constructed from $P'$ and $P''$ in linear time imply the following.

\begin{lemma}\label{le:bubble-merge}
$P^* = P' \merge P''$ can be computed in $O(|\mathcal S'|+|\mathcal S''|)$ time.
\end{lemma}

\subparagraph{Algorithm.}
Let $\cgraph{}{}$ be a  $2$-connected flat c-graph. Let $(D, \gamma)$ be a bond-carving decomposition of~$\dual(G)$ of width at most $\omega$ and let $\nu$ be a non-root bag of $D$. We denote by~$F_\nu$ the set of faces of $G$ that are dual to the vertices of $\dual(G)$ that are leaves of the subtree~$D_\nu$ of~$D$ rooted at $\nu$. Also, let $G_\nu$ be the embedded subgraph of $G$ induced by the edges of the faces in $F_\nu$. 
The \emph{interface graph}~$I_\nu$ of~$\nu$ is the subgraph of $G_\nu$ induced by the edges that are incident to a face of $G_\nu$ not in~$F_\nu$. The \emph{boundary} $B_\nu$ of $\nu$ is the vertex set of $I_\nu$. Note that, the edges of $I_\nu$ are dual to those in $\gamma(\nu)$. By~\cref{le:duality} and by the definition of bond-carving decomposition, we derive the next \mbox{observation about $I_\nu$.}

\begin{observation}\label{obs:boundary-cycle}
\mbox{The interface graph $I_\nu$ of $\nu$ is a cycle of length at most $\omega$.}
\end{observation}

Since $G$ is $2$-connected, by \cref{obs:boundary-cycle}, the vertices in $B_\nu$ have a natural (clockwise) circular ordering defined by cycle $I_\nu$, and $I_\nu$ bounds the {\em unique} face $f^\infty_\nu$ of $G_\nu$ not in $F_\nu$. Therefore, from now on, we regard $B_\nu$ as a cyclically-ordered set.

Let $P \in \NC(B_\nu)$  be an admissible partition  and let $\cgraph{\nu}{}$ be the flat c-graph obtained by restricting $\mathcal C$ to $G_\nu$. Also, let $\cgraph{\nu}{\diamond}$ be a super c-graph of $\mathcal C_\nu$ containing no saturating edges in the interior of $f^\infty_\nu$ and such that $G^\diamond_\nu$ is planar.

\begin{figure}[t]
\centering
\begin{subfigure}{.19\textwidth}
    \centering
  \includegraphics[page=10,height=\columnwidth]{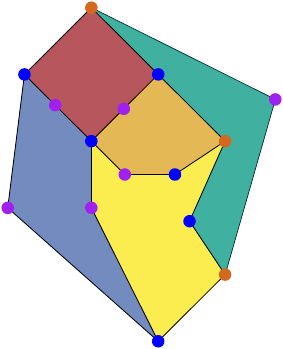}
  \subcaption{}
  \label{fi:realizable_a}
\end{subfigure}
\begin{subfigure}{.19\textwidth}
    \centering
  \includegraphics[page=11,height=\columnwidth]{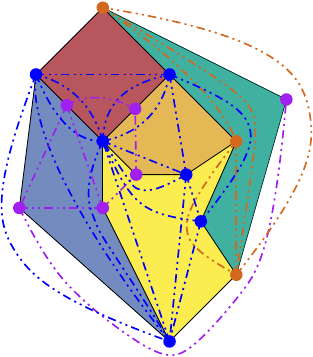}
  \subcaption{}
  \label{fi:realizable_b}
\end{subfigure}
\begin{subfigure}{.19\textwidth}
    \centering
  \includegraphics[page=12,height=\columnwidth]{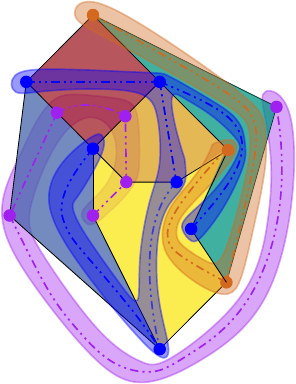}
  \subcaption{}
  \label{fi:realizable_c}
\end{subfigure}
\begin{subfigure}{.19\textwidth}
    \centering
  \includegraphics[page=6,height=\columnwidth]{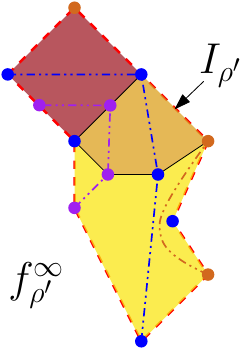}
  \subcaption{}
  \label{fi:realizable_d}
\end{subfigure}
\begin{subfigure}{.19\textwidth}
  \centering
  \includegraphics[page=7,height=\columnwidth]{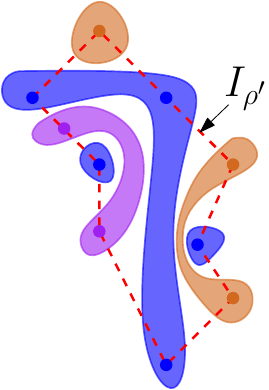}
  \subcaption{}
  \label{fi:realizable_e}
\end{subfigure}
\caption{(a) A flat c-graph $\calC{}(G,\calT{})$. (b) A super c-graph of $\calC{}$ containing all the candidate saturating edges. (c) A c-planar drawing of $\calC{}$ and the corresponding planar saturation. (d) A planar saturation of a c-graph, whose underlying graph is the graph $G_{\rho'}$ of the decomposition in \cref{fi:decomposition}, where no saturating edge lies in the interior of $f^\infty_{\rho'}$. (e) The admissible partition $P$ determined by the planar saturation in (d); sets of vertices of $I_{\rho'}$ belonging to the same cluster and connected by saturating edges in (d) form distinct parts in $P$ (enclosed by shaded regions).}\label{fi:realizable}
\end{figure}

\begin{definition} \label{def:realizes}
The c-graph $\mathcal C^\diamond_\nu$  \emph{realizes} $P$ if (refer to~\cref{fi:realizable}):
\begin{enumerate}[(a)]
\item \label{item:a} for every two vertices $u, v \in B_\nu$, we have that $u$ and $v$ belong to the same part $S_i \in P$ {\em if and only if} they are connected in $G^\diamond_\nu$ by paths of intra-cluster edges of the cluster the part $S_i$ belongs to,
\item \label{item:b} for each cluster $\mu$ in $\mathcal T$ such that $V_\mu \cap B_\nu \neq \emptyset$, all the vertices of $\mu$ in $G_\nu$ are connected to some vertex of $\mu$ in $B_\nu$ by paths of intra-cluster edges of $\mu$ in $G^\diamond_\nu$, and
\item \label{item:c} for each cluster $\mu$ in $\mathcal T$ such that $V_\mu \subseteq V(G_\nu) \setminus B_\nu$, all the vertices of $\mu$ are in $G_\nu$ and are connected by paths of intra-cluster edges of $\mu$ in $G^\diamond_\nu$.
\end{enumerate}
\end{definition}

A vertex $v$ of $G_\nu$ is \emph{dominated by} some part $S_i$ of $P$, if either $v \in S_i$ or $v$ is connected to some vertex of $S_i$ by paths of intra-cluster edges in $\mathcal C^\diamond_\nu$.
Note that, by \cref{item:a,item:b} of \cref{def:realizes}, for each part $S_i$ of $P$, the set of vertices of $G_\nu$ dominated by~$S_i$ induces a connected subgraph of $G^\diamond_\nu$. Also, by \cref{item:c},
cycle $I_\nu$ does not form a \emph{cluster separator}, that is, a cycle of $G$ such that the vertices of some cluster $\mu$ appear both in its interior and in its exterior, but not in it;
note that, in fact, this is a necessary condition for the existence of a c-planar drawing of $\mathcal C$. Thus, partition $P$ {\em``represents''} the internal-cluster connectivity in $\mathcal C^\diamond_\nu$ of the clusters whose vertices appear in $B_\nu$ {\em in a potentially positive instance}. Also, $P$ is \emph{realizable} by $\mathcal C_\nu$ if there exists a super c-graph $\cgraph{\nu}{\prime}$ of $\mathcal C_\nu$ that realizes~$P$ containing no saturating edges in the interior of $f^\infty_\nu$ and such that~$G^\prime_\nu$ is planar. 
From \cref{th:characterization} and the definition of realizable partition we have the following.

\begin{lemma}[Necessity]\label{le:necessity}
Let $\nu$ be a non-root bag of $D$. Then, $\mathcal C$ is c-planar {only if} there exists an admissible partition $P$ of $B_\nu$ that is realizable by $\mathcal C_\nu$.
\end{lemma}

We are going to exploit the next lemma, which holds for any bond-carving decomposition.

\begin{lemma}\label{le:interface-graphs}
Let $\rho'$ and $\rho''$ be the two children of the root $\rho$ of $D$.
Then, $I_{\rho'} = I_{\rho''}$.
\end{lemma}

\begin{proof}
Let $\mathcal L_{\rho'}$ and $\mathcal L_{\rho''}$ be the set of leaves of the subtree of $D$ rooted at $\rho'$ and at $\rho''$, respectively. 
Recall that, $\gamma(\rho') = (\mathcal L_{\rho'}, V \setminus \mathcal L_{\rho'})$ and that $\gamma(\rho'') = (\mathcal L_{\rho''}, V \setminus \mathcal L_{\rho''})$. Since 
$\mathcal L_{\rho'} = V \setminus \mathcal L_{\rho''}$ and $\mathcal L_{\rho''} = V \setminus \mathcal L_{\rho'}$, we have $\gamma(\rho')=\gamma(\rho'')$. 
Thus, cycles $I_{\rho'}$ and $I_{\rho''}$ consist of the edges of $G$ that are dual to the same edges of $\dual(G)$. This concludes the proof.
\end{proof}

\cref{le:necessity,le:interface-graphs} allow us to derive the following useful characterization.

\begin{theorem}[Characterization]\label{th:sufficiency}
The $2$-connected flat c-graph $\mathcal C(G, \mathcal T)$ is c-planar {if and only if} there exist admissible partitions $P' \in \NC(B_{\rho'})$, $P'' \in \NC(B_{\rho''})$ such that: 
\begin{enumerate}[(i)]
\item \label{con:1-suff} $P'$ and $P''$ are realizable by \cgraph{\rho'}{} and by \cgraph{\rho''}{}, respectively, and
\item \label{con:3-suff} no two distinct parts $S_i,S_j \in P^*$, with $P^* = P' \uplus P''$, belong to the same cluster of $\mathcal T$.
\end{enumerate}
\end{theorem}

\begin{proof}
We first prove the {\em only if} part. The necessity of \cref{con:1-suff} follows from~\cref{le:necessity}. 
For the necessity of \cref{con:3-suff}, suppose for a contradiction, that for any two realizable partitions $P' \in \NC(B_{\rho'})$ and $P'' \in \NC(B_{\rho''})$, it holds that $P^* = P' \uplus P''$ does not satisfy such a condition. Then, there is no set of saturating edges of $\mathcal C_{\rho'}$ and $\mathcal C_{\rho''}$, where none of these edges lies in the interior of $f^\infty_{\rho'}$ and of $f^\infty_{\rho''}$, respectively, that when added to $\mathcal C$ yields a c-connected c-graph $\cgraph{}{\diamond}$ with $G^\diamond$ planar. \mbox{Thus, by~\cref{th:characterization}, $\mathcal C$ is not c-planar, a contradiction.}

We now prove the {\em if part}. 
By~\cref{le:interface-graphs}, it holds $G = G_{\rho'} \cup G_{\rho''}$ and $I_{\rho'}=I_{\rho''} = G_{\rho'} \cap G_{\rho''}$.
Let $\mathcal C^\diamond_{\rho'}$ be a super c-graph of $\mathcal C_{\rho'}$ realizing $P'$ and let 
$\mathcal C^\diamond_{\rho''}$ be a super~c-graph~of~$\mathcal C_{\rho''}$ realizing $P''$; these c-graphs exist since \cref{con:1-suff} holds.
Let $\mathcal C^\diamond$ be the super c-graph of $\mathcal C$ obtained by augmenting $\mathcal C$ with the saturating edges of both $\mathcal C^\diamond_{\rho'}$ and $\mathcal C^\diamond_{\rho''}$.
Note that, $G^\diamond$ is planar.

We show that every cluster $\mu$ is connected in $\mathcal C^\diamond$, provided that \cref{con:3-suff} holds. This proves that $\mathcal C^\diamond$ is a c-connected super c-graph of $\mathcal C$, thus by \cref{con:3} of \cref{th:characterization}, c-graph $\cal C$ is c-planar. We distinguish two cases, based on whether some vertices of $\mu$ appear along cycle $I_{\rho'}=I_{\rho''}$ or not. Let $B = B_{\rho'} = B_{\rho''}$.

Consider first a cluster $\mu$ containing vertices in $B$. By \cref{item:b} of \cref{def:realizes}, we have that every vertex in $\mu$ is dominated by at least a part of $P'$ or of $P''$, i.e., they either belong to $B$ or they are connected by paths of intra-cluster edges in either~$\mathcal C^\diamond_{\rho'}$ or~$\mathcal C^\diamond_{\rho''}$ to a vertex in~$B$. Since, by \cref{con:3-suff} of the statement, there exists only one part $S_\mu \in P^*$ that contains vertices of cluster~$\mu$, we have that the different parts of $P'$ and of $P''$ containing vertices of~$\mu$ are joined together by the vertices of~$\mu$ in $B$. Therefore, the cluster $\mu$ is connected~in~$\mathcal C^\diamond$.
Finally, consider a cluster $\mu$ such that no vertex of $\mu$ belongs to $B$. Then, all the vertices of cluster $\mu$ only belong to either~$\mathcal C_{\rho'}$ or~$\mathcal C_{\rho''}$, by \cref{item:c} of \cref{def:realizes}. Suppose that  $\mu$ only belongs to $\mathcal C_{\rho'}$, the case when $\mu$ only belongs to $\mathcal C_{\rho''}$ is analogous. Since $\mathcal C^\diamond_{\rho'}$ realizes~$P'$, by \cref{item:c} of \cref{def:realizes}, all the vertices of $\mu$ in $G_\rho'$ are connected \mbox{by paths of intra-cluster edges.} Thus, cluster $\mu$ is connected in $\mathcal C^\diamond$, since it is connected in $\mathcal C^\diamond_{\rho'}$. This concludes the proof.
\end{proof}

We now present our main algorithmic tool.

\subparagraph{Algorithm~1.} Let $(D,\gamma)$ be a bond-carving decomposition of $\dual(G)$ of width $\omega$.
Let $\nu$ be a non-root bag of $D$, we denote by $R_\nu$ the set of all the admissible partitions of $B_\nu$ that are realizable by $\mathcal C_\nu$.
We process the bags of $D$ bottom-up and compute the following \emph{relevant information}, for each non-root bag $\nu$ of $D$:
{\bf 1.}~the set $R_\nu$, and
{\bf 2.}~for each admissible partition $P \in R_\nu$ and for each part $S_i \in P$, the number $count(S_i)$ of vertices of cluster $\mu$ belonging to $G_\nu$ that are dominated by $S_i$, where $\mu$ is the cluster~$S_i$ belongs to.

\begin{itemize}
\item If $\nu$ is a leaf bag of $D$, then  $G_\nu=I_\nu$ consists of the vertices and edges of a single face of $G$. Further, by \cref{obs:boundary-cycle}, graph $G_\nu$ is a cycle of \mbox{length at most $\omega$.} In this case, $R_\nu$ simply coincides with the set of all the admissible partitions of $B_\nu$.
Therefore, we can construct $R_\nu$ by enumerating all the possible at most $\CAT(\omega) \leq 2^{2\omega}$ non-crossing partitions of $B_\nu$ and by testing whether each such partition is good in $O(\omega)$ time. 
Further, for each $P \in R_\nu$, we can compute all counters $count(S_i)$ for every $S_i \in P$, in total $O(\omega)$ time, by visiting cycle $I_\nu$.

\item If $\nu$ is a non-leaf non-root bag of $D$, 
we have already computed the relevant information for the two children $\nu'$ and $\nu''$ of $\nu$. In the following way, we either detect that $\mathcal C$ does not satisfy \cref{con:3} of \cref{th:characterization} or construct the relevant information for $\nu$: 
\begin{enumerate}[(1)]
	\item \label{algo:1} Initialize $R_\nu = \emptyset$;
	\item \label{algo:2} For every pair of realizable admissible partitions $P' \in \mathcal R_{\nu'}$ and $P'' \in  R_{\nu''}$, perform the following operations:
	\begin{enumerate}[({2}a)]
		\item \label{algo:2.1} Compute $P^* = P' \uplus P''$  and compute the counters $count(S_i)$, for each $S_i \in P^*$, from the counters of the parts in $P' \cup P''$ whose union is $S_i$.
		\item \label{algo:2.2} If there exists some $S_i \in P^*$ such that $S_i \cap B_\nu = \emptyset$
		and $count(S_i)$ is smaller than the number of vertices in the cluster $S_i$ belongs to, then {\bf reject} the instance.
		\item \label{algo:2.3} Compute $P = P' \merge P''$ and add $P$ to $R_\nu$.
	\end{enumerate}
\end{enumerate}
\end{itemize}

\begin{remark}
Algorithm~1 rejects the instance at step (\ref{algo:2.2}), if $I_\mu$ forms a cluster separator. This property is independent of the specific generalized union $P^*$ considered at this step and implies that no $P^*$ (and, thus, no $P$ at step (\ref{algo:2.3})) can satisfy \cref{item:c} of \cref{def:realizes}.
\end{remark}

As the total number of pairs of partitions at step (\ref{algo:2}) is at most $(\CAT(\omega))^2$ and as $P^*$ and $P$ can be computed in $O(\omega)$ time, by~\cref{le:linear-union,le:bubble-merge}, we get the following.

\begin{lemma}\label{le:single-bag}
For each non-root bag $\nu$ of $D$, Algorithm~1 computes the relevant information for $\nu$ in~$O(2^{4 \omega + \log \omega})$ time, given the relevant information for its children.
\end{lemma}

\begin{proof}
Let $\nu'$ and $\nu''$ be the two children of $\nu$ in $D$.
We will first show the correctness of the algorithm and then argue about the running time.

Let $R^*_\nu$ be the set of all the admissible partitions of $B_\nu$ that are realizable by $\mathcal C_\nu$ and~let~$R_\nu$ be the set of all the admissible partitions of~$B_\nu$ computed by Algorithm~1. 
\mbox{We show $R_\nu$ = $R^*_\nu$.}

We first prove $R^*_\nu \subseteq R_\nu$.
Let $P_\nu$ be a realizable admissible partition in $R^*_\nu$. Since $P_\nu$ is realizable by $\mathcal C_\nu$, there exists a super c-graph $\cgraph{\nu}{*}$ of $\mathcal C_\nu$ that realizes $P_\nu$ containing no saturating edges in the interior of $f^\infty_\nu$ and such that $G^*_\nu$ is planar. 
Let $\cgraph{\nu'}{*}$ (resp.\ $\cgraph{\nu''}{*}$) be the super c-graph of $\cgraph{\nu'}{}$ (resp.\ of $\cgraph{\nu''}{}$) obtained by adding to 
${\mathcal C}_{\nu'}$ (resp.\ to ${\mathcal C}_{\nu''}$) all the saturating edges in $G^*_\nu$ laying in the interior of the faces of $G_{\nu'}$ (resp.\ of $G_{\nu''}$) that are also faces of $G_\nu$.
Clearly, c-graph $\cgraph{\nu'}{*}$ (resp.\ c-graph $\cgraph{\nu''}{*}$) contains no saturating edges in the interior of~$f^\infty_{\nu'}$ (resp.\ in the interior of $f^\infty_{\nu''}$), since such a face does not belong to $G_\nu$. Let $P'$ and $P''$ be the admissible partitions of $B_{\nu'}$ and of $B_{\nu''}$ realized by $\cgraph{\nu'}{*}$ and by $\cgraph{\nu''}{*}$, respectively. By hypothesis, we have $P' \in R_{\nu'}$ and $P'' \in R_{\nu''}$.
We show that when step (\ref{algo:2}) of Algorithm~1 considers partitions $P'$ and $P''$, it successfully adds $P_\nu$ to the set $R_\nu$.
It is clear by the construction of $P'$ and of $P''$ that $P_\nu = P' \merge P''$. Therefore, we only need to show that when the algorithm considers the pair $(P',P'')$, it does not reject the instance at step (\ref{algo:2.2}), and thus $P_\nu$ is added to $R_\nu$ at step (\ref{algo:2.3}).
Let $P^* = P' \uplus P''$, which is constructed at step (\ref{algo:2.1}) of the algorithm.
Suppose, for a contradiction, that $\cal C$ is rejected at step (\ref{algo:2.2}).
Then, there exists a part $S_i$ of $P^*$ such that $S_i \cap B_\nu = \emptyset$ and $count(S_i)$ is smaller than the number of vertices in the cluster $\mu$ the part $S_i$ belongs to. Therefore, the cluster $\mu$ contains vertices that belong to $G \setminus G_\nu$, which implies that $P_\nu$ cannot satisfy \cref{item:c} of \cref{def:realizes}, a contradiction. This concludes the proof of this direction.

We now prove $R_\nu \subseteq R^*_\nu$. Let $P$ be a partition in $R_\nu$ obtained from the partitions $P' \in R_{\nu'}$ and $P'' \in R_{\nu''}$ (selected at step (\ref{algo:2}) of the algorithm). 
We show that $P$ is realizable by~$\mathcal C_\nu$.

By the definition of realizable partition, there exists a super c-graph $\cgraph{\nu'}{*}$ (resp.\ $\cgraph{\nu''}{*}$) of $\cgraph{\nu'}{}$ (resp.\ of $\cgraph{\nu''}{}$) that realizes $P'$ (resp.\ $P''$)
containing no saturating edges in the interior of $f^\infty_{\nu'}$ (resp.\ of $f^\infty_{\nu''}$) and such that~$G^*_{\nu'}$ (resp.\ $G^*_{\nu''}$) is planar. 
Let $\cgraph{\nu}{*}$ be the super c-graph of $\cgraph{\nu}{}$ constructed by adding to $\mathcal C_{\nu}$ the saturating edges in $\mathcal C^*_{\nu'}$ and $\mathcal C^*_{\nu''}$. 
We show that the c-graph $\cgraph{\nu}{*}$ realizes~$P$, contains no saturating edges in the interior of $f^\infty_\nu$, and $G^*_\nu$ is planar.

First, we have that $G^*_\nu$ is planar, since $G^*_{\nu'}$  and $G^*_{\nu''}$ are planar and do not contain saturating edges in the interior of $f^\infty_{\nu'}$ and of $f^\infty_{\nu''}$, respectively. By the previous arguments, we also have that $f^\infty_{\nu}$ contains no saturating edges.

We show that \cref{item:a} of \cref{def:realizes} holds.
Recall that $P= P' \merge P''$. Let $S_i$ be a part of $P$ that also belongs to $P'$ or to $P''$.
Then, since ${\cal C}^*_{\nu'}$ and ${\cal C}^*_{\nu''}$ realize $P'$ and $P''$, respectively, the vertices of $S_i$ are connected by paths of intra-cluster edges in $C^*_\nu$ as they are connected by paths of intra-cluster edges in either ${\cal C}^*_{\nu'}$ or ${\cal C}^*_{\nu''}$, by \cref{item:a} of \cref{def:realizes}.
Otherwise, let $S_i$ be a part of $P$ that does not belong to either $P'$ or $P''$. Then, by the definition of bubble merge, the part $S_i$ is obtained by projecting onto $B_\nu$ the generalized union $P^* = P' \uplus P''$. 
Thus, $S_i$ is a subset of a part $S^*_i$ of $P^*$.
Also, the vertices in each of the parts of $P'$ and of $P''$ contributing to the creation of $S^*_i$ are connected by paths of intra-cluster edges in ${\cal C}^*_{\nu'}$ and ${\cal C}^*_{\nu''}$, respectively, by \cref{item:a} of \cref{def:realizes}.
Therefore, we have that the connectivity of such sets implies the connectivity of the elements of $S_i$ by paths of intra-cluster edges that connect at their shared vertices in $B_{\nu'} \cap B_{\nu''}$.
We show that \cref{item:b} of \cref{def:realizes} holds.
Suppose, for a contradiction, that there exists some
cluster $\mu$ whose vertices appear in $B_\nu$ such that there is at least a vertex of~$\mu$ in~$G_\nu$ that is not connected by a path of intra-cluster edges to some vertex of $\mu$ in $B_\nu$. Then, consider the part $S_i \in P^*$ that dominates this vertex, which exists since $P'$ and $P''$ are realizable by ${\cal C}_{\nu'}$ and by ${\cal C}_{\nu''}$, respectively. We have that $S_i \cap B_\nu = \emptyset$ and that $count(S_i)$ is smaller than the number of vertices of $\mu$. Thus, step (\ref{algo:2.2}) would reject the instance, and thus $P$ would not be added to $R_\nu$, a contradiction.
Finally, we show that \cref{item:c} of \cref{def:realizes} holds.
Suppose, for a contradiction, that there exists some
cluster $\mu$ whose vertices only belong to $V(G_\nu) \setminus B_\nu$ and that there exist two vertices $u$ and $v$ of $\mu$ in $G_\nu$ that are not connected by a path of intra-cluster edges in~$G^*_\nu$. Then, consider the part $S_i \in P^*$ that dominates $u$. Observe that, $S_i$ does not dominate $v$. Similarly to the proof of \cref{item:b}, we have that $S_i \cap B_\nu = \emptyset$ and that $count(S_i)$ is smaller than the number of vertices of $\mu$. Thus, step (\ref{algo:2.2}) would reject the instance, and thus $P$ would not be added to $R_\nu$, a contradiction.

We conclude by analyzing the running time.
Step (\ref{algo:2.1}) can be performed in linear time in the sum of the sizes of $B_{\nu'}$ and $B_{\nu''}$, since the generalized union $P^*$ can be computed in $O(|B_{\nu'}| + |B_{\nu''}|)$ time, by \cref{le:linear-union}, and since the size of $P^*$, and thus the number of counters to be updated, is in $O(|B_{\nu'}| + |B_{\nu''}|)$. 
Step (\ref{algo:2.2}) can also be done in linear time by the previous argument.
Step (\ref{algo:2.3}) can be performed in $O(|B_{\nu'}| + |B_{\nu''}|)$ time, by \cref{le:bubble-merge}.
Further, the number of pairs of realizable partitions considered at step (\ref{algo:2}) is bounded by $|\NC(B_{\nu'})| \cdot|\NC(B_{\nu''})|$, which is bounded by $2^{2(|B_{\nu'}|+|B_{\nu''}|)}$. Finally,  $|B_{\nu'}|\leq \omega$ and $|B_{\nu''}|\leq \omega$. Thus, Algorithm~1 runs in
$O(2^{4\omega}\omega) = O(2^{4\omega + \log \omega})$ time.
\end{proof}

By~\cref{le:single-bag} and since $D$ contains $O(n)$ bags, we have the following.

\begin{lemma}\label{le:all-bags}
Sets $R_{\rho'}$ and $R_{\rho''}$ can be computed in $O(2^{4\omega + \log \omega} n)$ time.
\end{lemma}

We obtain the next theorem by combining~\cref{le:all-bags} and~\cref{th:sufficiency}, where the additive $O(n^2)$ factor in the running time derives from the time needed to convert a carving decomposition of $\dual(G)$ into a bond-carving decomposition of the same width~\cite{DBLP:journals/combinatorica/SeymourT94}.

\begin{theorem}\label{th:2-connected-reduced}
{\sc C-Planarity Testing} can be solved in $O(2^{4 \omega + \log \omega} n + n^2)$ time for any $2$-connected $n$-vertex flat c-graph $\cgraph{}{}$, if a carving decomposition of $\dual(G)$ of width $\omega$ is provided.
\end{theorem}

We are finally ready to prove our main result. 

\subparagraph{Proof of~\cref{th:main}.} 
Let $(D,\gamma)$ be a carving decomposition of $\dual(G)$ of optimal width $\omega = \cw(\dual(G))$.
First, we apply~\cref{le:2connected} to $\mathcal C$ to obtain, in $O(n)$ time, a $2$-connected flat c-graph $\cgraph{}{\prime}$ equivalent to $\mathcal C$ and a corresponding carving decomposition $(D',\gamma')$ of width $\omega' \leq \max(\omega,4)$.
Then, we apply~\cref{th:2-connected-reduced} to test whether $\mathcal C'$ (and thus $\mathcal C$) is c-planar.
The running time follows from the running time of \cref{th:2-connected-reduced}, from the 
fact that $\omega'= O(\omega)$, $|V(G')| \in O(n)$, and that a carving decomposition of $\dual(G)$ of optimal width can be computed in $O(n^3)$ time~\cite{DBLP:journals/talg/GuT08,DBLP:conf/isaac/ThilikosSB00}. This concludes the proof of the theorem.
\smallskip

We remark that in the recent reduction presented by Patrignani and Cortese to convert any non-flat c-graph $\cgraph{}{}$ into an equivalent independent flat c-graph $\cgraph{}{\prime}$, the carving-width of $\dual(G^\prime)$ is within an $O(h)$ multiplicative factor from the carving-width of~$\dual(G)$, where $h$ is the height of $\mathcal T$. This is due to the fact that, by~\cite[Lemma~10]{DBLP:conf/gd/CorteseP18}, $G^\prime$ is a subdivision of $G$ (which implies that $\tw(G^\prime)=\tw(G)$) and that each inter-cluster edge of $G$ is replaced by a path of length at most $4h-4$ in $G^\prime$ (which implies that $\ell(G^\prime)=\ell(G)(4h-4)$). Therefore, by \cite{DBLP:conf/gd/CorteseP18} and by the results presented in this section, we immediately derive an FPT algorithm for the non-flat case parameterized by $h$ and the dual carving-width of $G$. In \cref{se:algorithm-non-flat}, we show how to drop the dependency on $h$, by suitably adapting the relevant concepts defined for the flat case so that Algorithm~1 can also be applied to the~non-flat~case.

\section{From Flat Instances to General Instances}\label{se:algorithm-non-flat}

In this section, we show how the concepts of partition, non-crossing partition, good partition, admissible partition, generalized union, projection, and bubble merge need to be modified so that the algorithm for flat c-graphs presented in \cref{se:algorithm} can also be applied to non-flat c-graphs.

In the context of non-flat c-graphs, we are going to use an extension of the notion of non-crossing partition that takes into account the inclusion between clusters.
Let $\mathcal Q = \{q_1, q_2, \dots, q_n\}$ be a ground set. In a \emph{recursive partition} $P = \{Q_1,\dots,Q_k\}$ of $\mathcal Q$, it holds that either $Q_i \cap Q_j = \emptyset$ or $Q_i \subset Q_j$, with $1 \leq i < j \leq k$. Observe that $k \leq 2|\mathcal Q| - 1$.
Let now $\mathcal S = (s_1, s_2, \dots, s_n)$ be a cyclically-ordered set. A recursive partition $P$ of $\mathcal S$ is \emph{crossing}, if there exist elements $a,c \in S_i$ and $b,d \in S_j$, with $S_i,S_j \in P$ and $S_i \cap S_j = \emptyset$ (i.e. neither~$S_i$ nor~$S_j$ are contained in one another), such that $a \prec_b c$ and $c \prec_d a$; and, it is \emph{non-crossing}, otherwise. We observe that the size of the set $\RNC(\mathcal S)$ of all the non-crossing recursive partitions of $\mathcal S$ coincides with the $\mathscr{CAT}(2n-1)$, as each non-crossing recursive partition is in one-to-one correspondence with a rooted ordered tree on $2n-1$ vertices whose leaves are the elements of $\mathcal S$~\cite{stanley_2015}. 

Given a c-graph $\cgraph{}{}$ and a set $S \subseteq V(G)$, the \emph{lowest common ancestor} of $S$ \emph{in} $\mathcal T$ is the deepest cluster $\mu \in \mathcal T$ that is encountered {\em by all} the vertices in $S$ on the path to the root. Also, we say that $S$ \emph{belongs} to the cluster $\mu$. A recursive partition $\{ S_1,\dots,S_k \}$ of~$V' \subset V(G)$ is \emph{good} if each part $S_i$ belongs to a non-root cluster.

\begin{figure}[tb!]
\centering
\begin{subfigure}{.35\textwidth}
    \centering
  \includegraphics[page=4,scale=0.4]{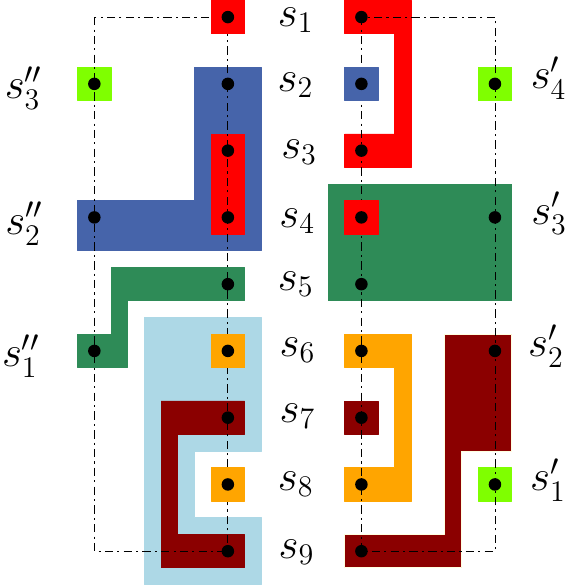}
  \subcaption{$P'$ and $P''$}
  \label{fi:auxGraph_a}
\end{subfigure}
\begin{subfigure}{.4\textwidth}
    \centering
  \includegraphics[page=5,scale=0.33]{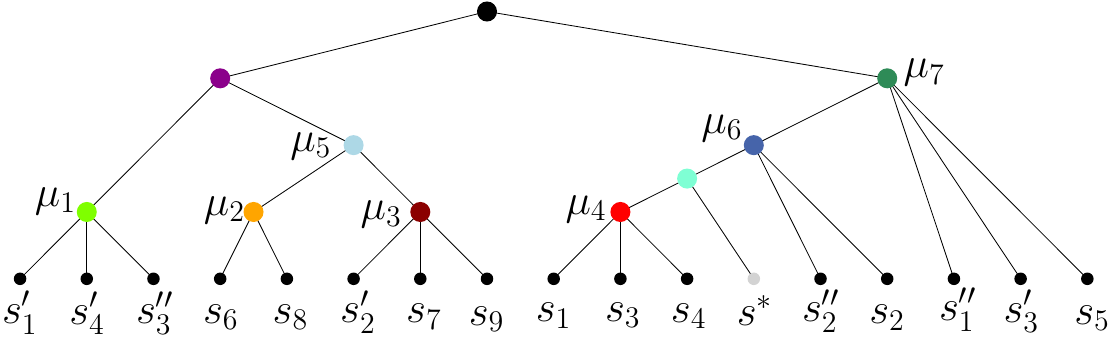}
  \subcaption{$\mathcal T$}
  \label{fi:auxGraph_b}
\end{subfigure}
\hfill
\begin{subfigure}{.2\textwidth}
    \centering
  \includegraphics[page=6,scale=0.33]{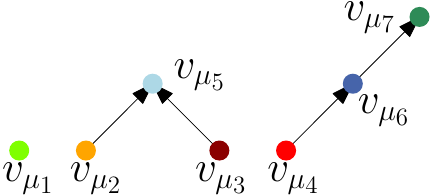}
  \subcaption{$A(P',P'')$}
  \label{fi:auxGraph_c}
\end{subfigure}
\caption{Illustration for the construction of the auxiliary graph $A(P',P'')$ of two good partitions $P'$ and $P''$ defined on the vertices of a c-graph $\cgraph{}{}$. Partitions $P'$ and $P''$ are, in fact, also non-crossing, and thus admissible.}\label{fi:auxGraph}
\end{figure}

In order to the extend the definition of generalized union to be appliable to pairs of good recursive partitions, we introduce the following auxiliary directed graph; refer to \cref{fi:auxGraph}. Given a pair $(P', P'')$ of good recursive partitions of ground sets $\mathcal Q' \subset V(G)$ and $\mathcal Q'' \subset V(G)$, respectively, we construct a graph $A(P', P'')$ as follows. The vertex set of $A(P',P'')$ contains a vertex $v_\mu$ if there exists a part $S \in P' \cup P''$ that belongs to the cluster $\mu$. The edge set of~$A(P', P'')$ contains an edge $v_\mu v_\nu$ directed from $v_\mu$ to $v_\nu$, if $\nu$ is an ancestor of $\mu$ in $\mathcal T$ and there exists no vertex $v_\tau$ in the vertex set of~$A(P', P'')$ such that the cluster $\tau$ belongs to the path connecting $\mu$ and $\nu$ in $\mathcal T$. Observe that, by construction, $A(P',P'')$ is a forest. Also, two parts $S', S'' \in P' \cup P''$ intersect only if there exists a directed path in $A(P', P'')$, possibly of length $0$, between the vertices corresponding to the clusters to which $S'$ and $S''$ belong. This holds since there exists a directed path between two vertices of~$A(P', P'')$ only if the corresponding clusters are one an ancestor of the other in $\mathcal T$. 

We are now ready to define the \emph{generalized union} $\uplus_R$ of $P'$ and $P''$ , which returns a partition $P^* = P' \uplus_R P''$ of $\mathcal Q' \cup \mathcal Q''$ obtained as follows; refer to \cref{fi:example_generalizedRecUnion}. Initialize $P^* = P' \cup P''$. Then, we perform the following two steps.
\begin{itemize}
	\item \textbf{Phase~1.} First, we visit the vertices of $A(P', P'')$. When we are at a vertex $v_\mu$, as long as there exist $Q_i,Q_j \in P^*$ such that $Q_i$ and $Q_j$ belong to the same cluster $\mu$ and $Q_i \cap Q_j \neq \emptyset$, we remove sets $Q_i$ and $Q_j$ from $P^*$ and add the set $Q_i \cup Q_j$ to $P^*$. Observe that, after this step is completed, any two parts $Q_i,Q_j \in P^*$ intersect only if $Q_i$ and $Q_j$ belong to different clusters. 
	\item \textbf{Phase~2.} Then, we again visit the vertices of $A(P', P'')$, but according to a topological ordering of such vertices. When we are at a vertex $v_\mu$, as long as there exist $Q_i,Q_j \in P^*$ such that:
	\begin{enumerate}[(a)]
	 	\item $Q_i \cap Q_j \neq \emptyset$, and
	 	\item $Q_i$ belongs to the cluster $\mu$, $Q_j$ belongs to the cluster $\nu$, and $v_\mu v_\nu$ is an edge of $A(P', P'')$,
	 \end{enumerate}
	 we remove the part $Q_j$ from $P^*$ and \mbox{add the part $Q_i \cup Q_j$ to $P^*$.}
\end{itemize}

Clearly, by the above definition, the generalized union of $P'$ and $P''$ can be computed in $O\big((|\mathcal Q'|+|\mathcal Q''|)^2\big)$ time.

\begin{figure}
\centering
\begin{subfigure}{.32\textwidth}
    \centering
  \includegraphics[page=8,scale=0.35]{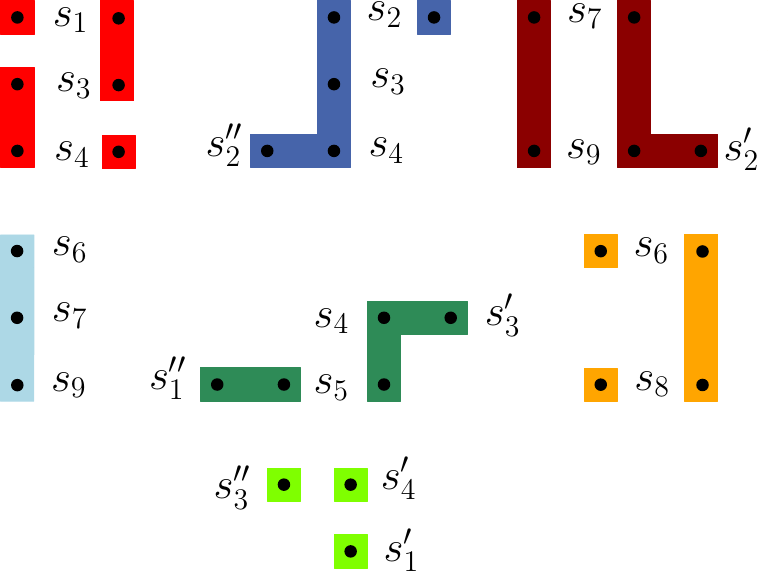}
  \subcaption{$P' \cup P''$}
  \label{fi:example_generalizedRecUnion_a}
\end{subfigure}
\hfill
\begin{subfigure}{.32\textwidth}
    \centering
  \includegraphics[page=10,scale=0.35]{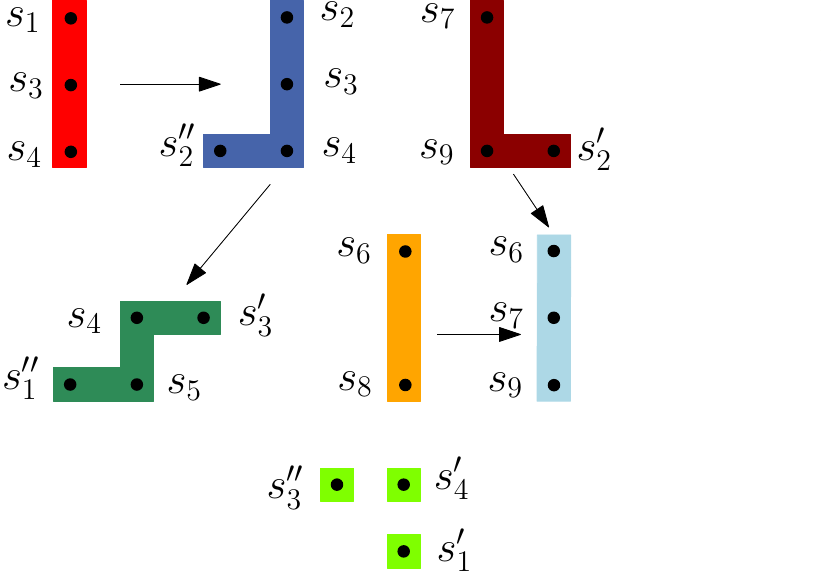}
  \subcaption{Phase~1}
  \label{fi:example_generalizedRecUnion_b}
\end{subfigure}
\begin{subfigure}{.32\textwidth}
    \centering
  \includegraphics[page=11,scale=0.35]{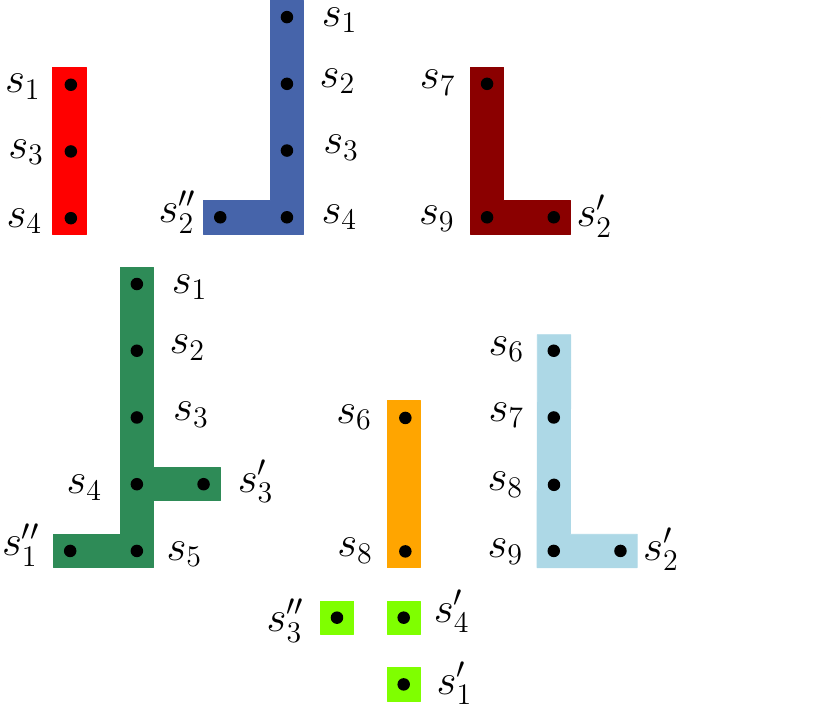}
  \subcaption{Phase~2}
  \label{fi:example_generalizedRecUnion_c}
\end{subfigure}
\caption{Illustrations for the construction of the generalized union of the two admissible recursive partitions $P'$ and $P''$ illustrated in \cref{fi:auxGraph_a}. Different parts are represented by shaded polygons.
The arrows in (b) show how the parts obtained in Phase~1 are merged in Phase~2, according to the auxiliary graph $A(P',P'')$, to obtain (c).
}\label{fi:example_generalizedRecUnion}
\end{figure}

The definition of \emph{admissible recursive partition} and of \emph{projection} of a recursive partition \emph{onto} a set are identical to the corresponding definitions in the non-recursive setting.

In the same way as an admissible partition $P$ can be naturally associated with a cycle-star $G(P)$, an admissible recursive partition $P'$ of a cyclically-ordered set $\mathcal S$ can be naturally associated with a $2$-connected plane graph $G_R(P')$, as follows.  
The outer face of $G_R(P')$ is a cycle $C(P')$ whose vertices are the elements in $\mathcal S$ and the clockwise order in which they appear along $C(P')$ is the same as in $\mathcal S$. Also, for each part $S_i \in P'$, graph $G_R(P')$ contains a vertex $v_i$ in the interior of $C(P')$.
We further add the following edges to the graph $G_R(P')$.
First, we add an edge connecting two vertices $v_i$ and $v_j$ if and only if $S_j$ is the smallest part in $P'$ such that $S_i \subset S_j$.
Then, we add an edge connecting a vertex $v \in \mathcal S$ with a vertex $v_i$ if and only if $S_i$ is the smallest part in $P'$ containing $v$.
We say that $G_R(P')$ is the \emph{cycle-tree} associated with $P'$; see \cref{fi:recursive_projection_c} for an example.

The definition of \emph{bubble merge} $\merge$ of two admissible recursive partition is identical to the corresponding definition in the non-recursive setting where, however, the operator $\uplus_R$ is used instead of the operator $\uplus$, and the role played by cycle-stars is now assumed by the cycle-trees associated with the input admissible recursive partitions; refer to \cref{fi:recursive_projection}. Since the cycle-tree~$G_R(P)$ associated with an admissible recursive partition on a cyclically-ordered set~$\cal S$ can be constructed in $O(|{\cal S}|^2)$ time and it has size linear in $|\cal S|$, the bubble merge of two admissible recursive partitions $P'$ and $P''$ of cyclically-ordered sets $\mathcal{S}'$ and $\mathcal{S}''$, respectively, can be performed in $O\big((|{\cal S}'| + |{\cal S}''|)^2\big)$ time.

The definition of \emph{realizable} (admissible partition) and of \emph{dominated by} are identical to the corresponding definitions in the non-recursive setting, where we exploit the definition of {\em belongs} given above (i.e., a part belongs to a cluster $\mu$, if $\mu$ is the lowest common ancestor of the part).

With the modification to the relevant concepts given above, Algorithm~1 can be applied to compute the relevant information for each non-leaf non-root bag $\nu$ of the bond-carving decomposition $(D,\gamma)$ of $\dual(G)$. However, in the non-flat case, we need to compute the relevant information for the leaf bags in a slightly different way than in the flat case.
Recall that, if $\nu$ is a leaf bag of $D$, then $G_\nu=I_\nu$ consists of the vertices and edges of a single face of~$G$. Further, by \cref{obs:boundary-cycle}, graph $G_\nu$ is a cycle of \mbox{length at most $\omega$.} In this case, the sets $R_\nu$ of all the admissible recursive partitions of $B_\nu$ that are realizable by $G_\nu$ simply coincides with the set of all the admissible recursive partitions of $B_\nu$.
Therefore, we can construct the set~$R_\nu$ by enumerating all the possible at most $\CAT(2\omega-1) \leq 2^{4\omega-2}$ non-crossing recursive partitions of $B_\nu$ and by testing in $O(\omega)$ time whether each of such partitions is good. 
Further, for each $P \in R_\nu$, we can compute all counters $count(S_i)$ for every $S_i \in P$ in total $O(\omega)$ time, by visiting the cycle-tree $G_R(P)$ associated with $P$. 

Since the generalized union and, thus, the bubble merge operators take quadratic time, the computation of the relevant information at each non-leaf node of $D$ can be done in $O(2^{4\omega-2}2^{4\omega-2}\omega^2)=O(4^{4 \omega + \log \omega})$ time, given the relevant information for its two children. Therefore, we obtain the following main result.

\begin{theorem}\label{th:non-flat-algorithm}
{\sc C-Planarity Testing} can be solved in $O(4^{4 \omega + \log \omega} n + n^2)$ time for any $n$-vertex embedded non-flat c-graph $\cgraph{}{}$, where $\omega$ is the carving-width of $\dual(G)$, if a carving decomposition of $\dual(G)$ of width $\omega$ is provided, and in $O(4^{4 \omega + \log \omega} n + n^3)$ time, otherwise.
\end{theorem}

\begin{figure}[t]
\centering
\hfill
\begin{subfigure}{.32\textwidth}
    \centering
  \includegraphics[page=12,scale=.6]{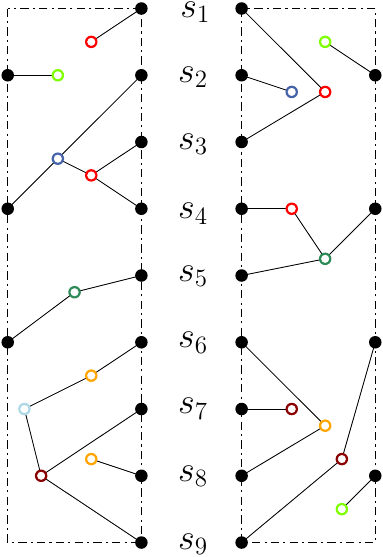}
  \subcaption{$G_R(P')$ and $G_R(P'')$}
  \label{fi:recursive_projection_a}
\end{subfigure}
\hfil
\begin{subfigure}{.32\textwidth}
    \centering
  \includegraphics[page=13,scale=.6]{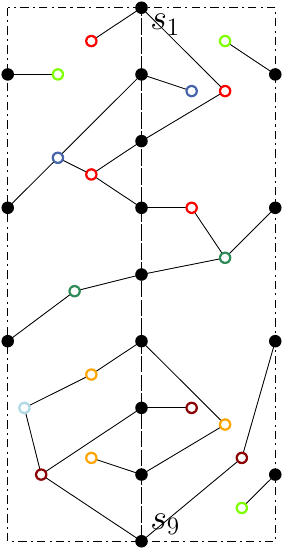}
  \subcaption{$H$}
  \label{fi:recursive_projection_b}
\end{subfigure}
\begin{subfigure}{.32\textwidth}
    \centering
  \includegraphics[page=14,scale=.6]{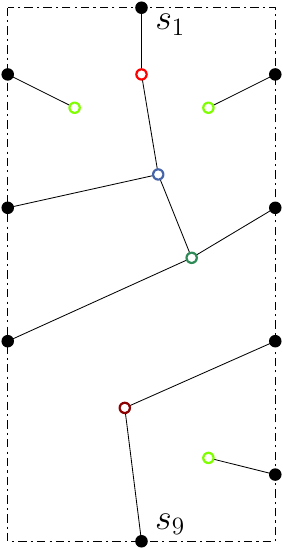}
  \subcaption{$G_R(P^*)$}
  \label{fi:recursive_projection_c}
\end{subfigure}
\caption{Illustrations for the construction of the bubble merge of the two admissible recursive partitions $P'$ and $P''$ illustrated in \cref{fi:auxGraph_a}.}\label{fi:recursive_projection}
\end{figure}

\section{Graph-Width Parameters Related to the Dual Carving-Width}\label{se:related-params}

In this section, we discuss implications of our algorithm for instances of bounded embedded-width and of bounded dual cut-width.

\subparagraph{Embedded-width.} A tree decomposition of an embedded graph $G$ \emph{respects} the embedding of~$G$ if, for every face $f$ of $G$, at least one bag contains all the vertices of $f$~\cite{BELW}. 
The \emph{embedded-width} $\emw(G)$ of $G$ is the minimum width of any of its tree decompositions that respect the embedding of $G$.  
For consistency with other graph-width parameters, in the original definition of this width measure~\cite{BELW} the vertices of the outer face are not required to be in some bag. Here, we adopt the variant presented in~\cite{degg-sfefcg-conf-17}, where the tree decomposition must also include \mbox{a bag containing the outer face.}
We have the following.

\begin{lemma}\label{le:relations-embedded}
Let $G$ be an embedded graph. Then, $\cw(\dual(G)) \leq {\emw}^2(G) + 2{\emw}(G)$.
\end{lemma}

\begin{proof}
Recall that the dual $\dual(G)$ of $G$ has maximum degree $\ell(G)$ and maximum face size~$\Delta(G)$.
It is well-known that the tree-width $\tw(G)$ of $G$ and the tree-width $\tw(\dual(G))$ of~$\dual(G)$ satisfy the relation: $\tw(\dual(G)) \leq \tw(G) + 1$~\cite{DBLP:journals/endm/BouchitteMT01}. Also, for any graph $H$, the carving-width $\cw(H)$ of $H$ satisfies the relation: $\cw(H) \leq \Delta(H)(\tw(H)+1)$~\cite{DBLP:journals/ijcga/BiedlV13}.
Therefore, we have $\cw(\dual(G)) \leq \ell(G)(\tw(G)+2)$.
Finally, the embedded-width ${\emw}(G)$ of $G$ satisfies the relations: ${\emw}(G)\geq \ell$ and ${\emw}(G)\geq \tw(G)$, by definition~\cite{BELW,degg-sfefcg-conf-17}.
Combining the above inequalities, we get the stated bound for the carving-width of $\dual(G)$. 
\end{proof}

\subparagraph{Cut-width.} Let $\pi$ be a linear order of the vertex set of a graph $G=(V,E)$. 
By splitting~$\pi$ into two linear orders $\pi_1$ and $\pi_2$ such that $\pi$ is the concatenation of $\pi_1$ and $\pi_2$, we define a \emph{cut} of $\pi$. 
The \emph{width} of this cut is the number of edges between a vertex in $\pi_1$ and \mbox{a vertex in $\pi_2$.}
The \emph{width} of $\pi$ is the maximum width over all its possible cuts. Finally, the \emph{cut-width} of $G$ is the minimum width over all the possible linear orders of $V$. The \emph{dual cut-width} is the cut-width of the dual~of~$G$.

The following relationship between cut-width and carving-width has been proved in~\cite{cgp}.

\begin{theorem}[Theorem 4.3, \cite{cgp}]\label{le:relations-cut}
\mbox{The carving-width of $G$ is at most twice its cut-width.}
\end{theorem}

By \cref{le:relations-embedded,le:relations-cut}, we have that single-parameter FPT algorithms also exist with respect to the embedded-width and to the dual cut-width of the underlying graph.

\section{Conclusions}\label{se:conclusion}

In this paper, we studied the {\sc C-Planarity Testing} problem for c-graphs with a prescribed combinatorial embedding. We showed that the problem is polynomial-time solvable when the dual carving-width of the underlying graph of the input c-graph is bounded. In particular, this addresses a question we posed in~\cite{degg-sfefcg-conf-17}, regarding the existence of notable graph-width parameters such that the {\sc C-Planarity Testing} problem is fixed-parameter tractable with respect to a single one of them. Namely, we answer this question in the affirmative when the parameters are the embedded-width of the underlying graph, and the carving-width and cut-width of its planar dual.

\bibliography{bibliography}

\begin{thebibliography}{10}

\bibitem{DBLP:conf/wg/AdlerBRRTV10}
Isolde Adler, Binh{-}Minh Bui{-}Xuan, Yuri Rabinovich, Gabriel Renault,
  Jan~Arne Telle, and Martin Vatshelle.
\newblock On the boolean-width of a graph: Structure and applications.
\newblock In Dimitrios~M. Thilikos, editor, {\em {WG} 2010}, volume 6410 of
  {\em LNCS}, pages 159--170, 2010.
\newblock URL: \url{https://doi.org/10.1007/978-3-642-16926-7\_16}, \href
  {http://dx.doi.org/10.1007/978-3-642-16926-7\_16}
  {\path{doi:10.1007/978-3-642-16926-7\_16}}.

\bibitem{DBLP:conf/soda/AkitayaFT18}
Hugo~A. Akitaya, Radoslav Fulek, and Csaba~D. T{\'{o}}th.
\newblock Recognizing weak embeddings of graphs.
\newblock In Artur Czumaj, editor, {\em {SODA}~'18}, pages 274--292. {SIAM},
  2018.
\newblock URL: \url{https://doi.org/10.1137/1.9781611975031.20}, \href
  {http://dx.doi.org/10.1137/1.9781611975031.20}
  {\path{doi:10.1137/1.9781611975031.20}}.

\bibitem{DBLP:journals/cj/AngeliniL16}
Patrizio Angelini and Giordano {Da Lozzo}.
\newblock {SEFE} = {C}-{P}lanarity?
\newblock {\em Comput. J.}, 59(12):1831--1838, 2016.
\newblock URL: \url{https://doi.org/10.1093/comjnl/bxw035}, \href
  {http://dx.doi.org/10.1093/comjnl/bxw035} {\path{doi:10.1093/comjnl/bxw035}}.

\bibitem{DBLP:conf/isaac/AngeliniL16}
Patrizio Angelini and Giordano {Da Lozzo}.
\newblock Clustered planarity with pipes.
\newblock {\em Algorithmica}, 81(6):2484--2526, 2019.
\newblock URL: \url{https://doi.org/10.1007/s00453-018-00541-w}, \href
  {http://dx.doi.org/10.1007/s00453-018-00541-w}
  {\path{doi:10.1007/s00453-018-00541-w}}.

\bibitem{addf-sprepg-17}
Patrizio Angelini, Giordano {Da Lozzo}, Giuseppe {Di Battista}, and Fabrizio
  Frati.
\newblock Strip planarity testing for embedded planar graphs.
\newblock {\em Algorithmica}, 77(4):1022--1059, 2017.
\newblock URL: \url{https://doi.org/10.1007/s00453-016-0128-9}, \href
  {http://dx.doi.org/10.1007/s00453-016-0128-9}
  {\path{doi:10.1007/s00453-016-0128-9}}.

\bibitem{DBLP:journals/comgeo/AngeliniLBFPR15}
Patrizio Angelini, Giordano {Da Lozzo}, Giuseppe {Di Battista}, Fabrizio Frati,
  Maurizio Patrignani, and Vincenzo Roselli.
\newblock Relaxing the constraints of clustered planarity.
\newblock {\em Comput. Geom.}, 48(2):42--75, 2015.
\newblock URL: \url{https://doi.org/10.1016/j.comgeo.2014.08.001}, \href
  {http://dx.doi.org/10.1016/j.comgeo.2014.08.001}
  {\path{doi:10.1016/j.comgeo.2014.08.001}}.

\bibitem{DBLP:journals/jgaa/AngeliniLBFPR17}
Patrizio Angelini, Giordano {Da Lozzo}, Giuseppe {Di Battista}, Fabrizio Frati,
  Maurizio Patrignani, and Ignaz Rutter.
\newblock Intersection-link representations of graphs.
\newblock {\em J. Graph Algorithms Appl.}, 21(4):731--755, 2017.
\newblock URL: \url{https://doi.org/10.7155/jgaa.00437}, \href
  {http://dx.doi.org/10.7155/jgaa.00437} {\path{doi:10.7155/jgaa.00437}}.

\bibitem{ClusteredLevel15}
Patrizio Angelini, Giordano {Da Lozzo}, Giuseppe {Di Battista}, Fabrizio Frati,
  and Vincenzo Roselli.
\newblock The importance of being proper: (in clustered-level planarity and
  {T}-level planarity).
\newblock {\em Theor. Comput. Sci.}, 571:1--9, 2015.
\newblock URL: \url{https://doi.org/10.1016/j.tcs.2014.12.019}, \href
  {http://dx.doi.org/10.1016/j.tcs.2014.12.019}
  {\path{doi:10.1016/j.tcs.2014.12.019}}.

\bibitem{DBLP:journals/dcg/AngeliniFK11}
Patrizio Angelini, Fabrizio Frati, and Michael Kaufmann.
\newblock Straight-line rectangular drawings of clustered graphs.
\newblock {\em Discrete {\&} Computational Geometry}, 45(1):88--140, 2011.
\newblock URL: \url{https://doi.org/10.1007/s00454-010-9302-z}, \href
  {http://dx.doi.org/10.1007/s00454-010-9302-z}
  {\path{doi:10.1007/s00454-010-9302-z}}.

\bibitem{DBLP:journals/jgaa/AthenstadtC17}
Jan~Christoph Athenst{\"{a}}dt and Sabine Cornelsen.
\newblock Planarity of overlapping clusterings including unions of two
  partitions.
\newblock {\em J. Graph Algorithms Appl.}, 21(6):1057--1089, 2017.
\newblock URL: \url{https://doi.org/10.7155/jgaa.00450}, \href
  {http://dx.doi.org/10.7155/jgaa.00450} {\path{doi:10.7155/jgaa.00450}}.

\bibitem{b-dppIIItcep-98}
T.~Biedl.
\newblock Drawing planar partitions {III}: {Two Constrained Embedding
  Problems}.
\newblock Tech. {Report} RRR 13-98, Rutcor Research Report, 1998.

\bibitem{DBLP:journals/ijcga/BiedlV13}
Therese~C. Biedl and Martin Vatshelle.
\newblock The point-set embeddability problem for plane graphs.
\newblock {\em Int. J. Comput. Geometry Appl.}, 23(4-5):357--396, 2013.
\newblock URL: \url{https://doi.org/10.1142/S0218195913600091}, \href
  {http://dx.doi.org/10.1142/S0218195913600091}
  {\path{doi:10.1142/S0218195913600091}}.

\bibitem{BlasiusR16}
Thomas Bl{\"{a}}sius and Ignaz Rutter.
\newblock A new perspective on clustered planarity as a combinatorial embedding
  problem.
\newblock {\em Theor. Comput. Sci.}, 609:306--315, 2016.
\newblock URL: \url{https://doi.org/10.1016/j.tcs.2015.10.011}, \href
  {http://dx.doi.org/10.1016/j.tcs.2015.10.011}
  {\path{doi:10.1016/j.tcs.2015.10.011}}.

\bibitem{BELW}
Glencora Borradaile, Jeff Erickson, Hung Le, and Robbie Weber.
\newblock Embedded-width: {A} variant of treewidth for plane graphs, 2017.
\newblock URL: \url{http://arxiv.org/abs/1703.07532}, \href
  {http://arxiv.org/abs/1703.07532} {\path{arXiv:1703.07532}}.

\bibitem{DBLP:journals/endm/BouchitteMT01}
Vincent Bouchitt{\'{e}}, Fr{\'{e}}d{\'{e}}ric Mazoit, and Ioan Todinca.
\newblock Treewidth of planar graphs: connections with duality.
\newblock {\em ENDM}, 10:34--38, 2001.
\newblock URL: \url{https://doi.org/10.1016/S1571-0653(04)00353-1}, \href
  {http://dx.doi.org/10.1016/S1571-0653(04)00353-1}
  {\path{doi:10.1016/S1571-0653(04)00353-1}}.

\bibitem{DBLP:conf/gd/BrandenburgEGKLM03}
Franz{-}Josef Brandenburg, David Eppstein, Michael~T. Goodrich, Stephen~G.
  Kobourov, Giuseppe Liotta, and Petra Mutzel.
\newblock Selected open problems in graph drawing.
\newblock In Giuseppe Liotta, editor, {\em {GD}~'03}, volume 2912 of {\em
  LNCS}, pages 515--539. Springer, 2003.
\newblock URL: \url{https://doi.org/10.1007/978-3-540-24595-7_55}, \href
  {http://dx.doi.org/10.1007/978-3-540-24595-7_55}
  {\path{doi:10.1007/978-3-540-24595-7_55}}.

\bibitem{cdfk-atcpefcg-14}
Markus Chimani, Giuseppe {Di Battista}, Fabrizio Frati, and Karsten Klein.
\newblock Advances on testing c-planarity of embedded flat clustered graphs.
\newblock In Christian~A. Duncan and Antonios Symvonis, editors, {\em
  {GD}~'14}, volume 8871 of {\em LNCS}, pages 416--427. Springer, 2014.
\newblock URL: \url{https://doi.org/10.1007/978-3-662-45803-7_35}, \href
  {http://dx.doi.org/10.1007/978-3-662-45803-7_35}
  {\path{doi:10.1007/978-3-662-45803-7_35}}.

\bibitem{ck-ssscp-12}
Markus Chimani and Karsten Klein.
\newblock Shrinking the search space for clustered planarity.
\newblock In Walter Didimo and Maurizio Patrignani, editors, {\em {GD}~'12},
  volume 7704 of {\em LNCS}, pages 90--101. Springer, 2012.
\newblock URL: \url{https://doi.org/10.1007/978-3-642-36763-2_9}, \href
  {http://dx.doi.org/10.1007/978-3-642-36763-2_9}
  {\path{doi:10.1007/978-3-642-36763-2_9}}.

\bibitem{DBLP:journals/siamcomp/CorneilR05}
Derek~G. Corneil and Udi Rotics.
\newblock On the relationship between clique-width and treewidth.
\newblock {\em {SIAM} J. Comput.}, 34(4):825--847, 2005.
\newblock URL: \url{https://doi.org/10.1137/S0097539701385351}, \href
  {http://dx.doi.org/10.1137/S0097539701385351}
  {\path{doi:10.1137/S0097539701385351}}.

\bibitem{CornelsenW06}
Sabine Cornelsen and Dorothea Wagner.
\newblock Completely connected clustered graphs.
\newblock {\em J. Discrete Algorithms}, 4(2):313--323, 2006.
\newblock URL: \url{https://doi.org/10.1016/j.jda.2005.06.002}, \href
  {http://dx.doi.org/10.1016/j.jda.2005.06.002}
  {\path{doi:10.1016/j.jda.2005.06.002}}.

\bibitem{DBLP:conf/compgeom/CorteseB05}
Pier~Francesco Cortese and Giuseppe {Di Battista}.
\newblock Clustered planarity.
\newblock In Joseph S.~B. Mitchell and G{\"{u}}nter Rote, editors, {\em
  {S}o{CG}~'05}, pages 32--34. {ACM}, 2005.
\newblock URL: \url{http://doi.acm.org/10.1145/1064092.1064093}, \href
  {http://dx.doi.org/10.1145/1064092.1064093}
  {\path{doi:10.1145/1064092.1064093}}.

\bibitem{cdfpp-cccg-06}
Pier~Francesco Cortese, Giuseppe {Di Battista}, Fabrizio Frati, Maurizio
  Patrignani, and Maurizio Pizzonia.
\newblock C-planarity of c-connected clustered graphs.
\newblock {\em J. Graph Algorithms Appl.}, 12(2):225--262, 2008.
\newblock URL: \url{http://jgaa.info/accepted/2008/Cortese+2008.12.2.pdf}.

\bibitem{DBLP:journals/dm/CorteseBPP09}
Pier~Francesco Cortese, Giuseppe {Di Battista}, Maurizio Patrignani, and
  Maurizio Pizzonia.
\newblock On embedding a cycle in a plane graph.
\newblock {\em Discrete Mathematics}, 309(7):1856--1869, 2009.
\newblock URL: \url{https://doi.org/10.1016/j.disc.2007.12.090}, \href
  {http://dx.doi.org/10.1016/j.disc.2007.12.090}
  {\path{doi:10.1016/j.disc.2007.12.090}}.

\bibitem{DBLP:conf/gd/CorteseP18}
Pier~Francesco Cortese and Maurizio Patrignani.
\newblock Clustered planarity = flat clustered planarity.
\newblock In Therese~C. Biedl and Andreas Kerren, editors, {\em {GD} 2018},
  volume 11282 of {\em LNCS}, pages 23--38. Springer, 2018.
\newblock URL: \url{https://doi.org/10.1007/978-3-030-04414-5\_2}, \href
  {http://dx.doi.org/10.1007/978-3-030-04414-5\_2}
  {\path{doi:10.1007/978-3-030-04414-5\_2}}.

\bibitem{DBLP:journals/jgaa/LozzoBFP18}
Giordano {Da Lozzo}, Giuseppe {Di Battista}, Fabrizio Frati, and Maurizio
  Patrignani.
\newblock Computing nodetrix representations of clustered graphs.
\newblock {\em J. Graph Algorithms Appl.}, 22(2):139--176, 2018.
\newblock URL: \url{https://doi.org/10.7155/jgaa.00461}, \href
  {http://dx.doi.org/10.7155/jgaa.00461} {\path{doi:10.7155/jgaa.00461}}.

\bibitem{degg-sfefcg-conf-17}
Giordano {Da Lozzo}, David Eppstein, Michael~T. Goodrich, and Siddharth Gupta.
\newblock Subexponential-time and {FPT} algorithms for embedded flat clustered
  planarity.
\newblock In Andreas Brandst{\"{a}}dt, Ekkehard K{\"{o}}hler, and Klaus Meer,
  editors, {\em {WG} 2018}, volume 11159 of {\em LNCS}, pages 111--124.
  Springer, 2018.
\newblock URL: \url{https://doi.org/10.1007/978-3-030-00256-5\_10}, \href
  {http://dx.doi.org/10.1007/978-3-030-00256-5\_10}
  {\path{doi:10.1007/978-3-030-00256-5\_10}}.

\bibitem{DBLP:conf/latin/Dahlhaus98}
Elias Dahlhaus.
\newblock A linear time algorithm to recognize clustered graphs and its
  parallelization.
\newblock In Claudio~L. Lucchesi and Arnaldo~V. Moura, editors, {\em
  {LATIN}~'98}, volume 1380 of {\em LNCS}, pages 239--248. Springer, 1998.
\newblock URL: \url{https://doi.org/10.1007/BFb0054325}, \href
  {http://dx.doi.org/10.1007/BFb0054325} {\path{doi:10.1007/BFb0054325}}.

\bibitem{DBLP:conf/gd/BattistaDM01}
Giuseppe {Di Battista}, Walter Didimo, and A.~Marcandalli.
\newblock Planarization of clustered graphs.
\newblock In Petra Mutzel, Michael J{\"{u}}nger, and Sebastian Leipert,
  editors, {\em {GD}~'01}, volume 2265 of {\em LNCS}, pages 60--74. Springer,
  2001.
\newblock URL: \url{https://doi.org/10.1007/3-540-45848-4_5}, \href
  {http://dx.doi.org/10.1007/3-540-45848-4_5}
  {\path{doi:10.1007/3-540-45848-4_5}}.

\bibitem{df-ectefgsf-13}
Giuseppe {Di Battista} and Fabrizio Frati.
\newblock Efficient c-planarity testing for embedded flat clustered graphs with
  small faces.
\newblock {\em J. Graph Algorithms Appl.}, 13(3):349--378, 2009.
\newblock URL:
  \url{http://jgaa.info/accepted/2009/DiBattistaFrati2009.13.3.pdf}.

\bibitem{DBLP:journals/jgaa/DidimoGL08}
Walter Didimo, Francesco Giordano, and Giuseppe Liotta.
\newblock Overlapping cluster planarity.
\newblock {\em J. Graph Algorithms Appl.}, 12(3):267--291, 2008.
\newblock URL:
  \url{http://jgaa.info/accepted/2008/DidimoGiordanoLiotta2008.12.3.pdf}.

\bibitem{fce-pcg-95}
Qing{-}Wen Feng, Robert~F. Cohen, and Peter Eades.
\newblock Planarity for clustered graphs.
\newblock In Paul~G. Spirakis, editor, {\em {ESA}'95}, volume 979 of {\em
  LNCS}, pages 213--226. Springer, 1995.
\newblock URL: \url{https://doi.org/10.1007/3-540-60313-1_145}, \href
  {http://dx.doi.org/10.1007/3-540-60313-1_145}
  {\path{doi:10.1007/3-540-60313-1_145}}.

\bibitem{DBLP:conf/sofsem/ForsterB04}
Michael Forster and Christian Bachmaier.
\newblock Clustered level planarity.
\newblock In Peter van Emde~Boas, Jaroslav Pokorn{\'{y}}, M{\'{a}}ria
  Bielikov{\'{a}}, and Julius Stuller, editors, {\em {SOFSEM}~'04}, volume 2932
  of {\em LNCS}, pages 218--228. Springer, 2004.
\newblock URL: \url{https://doi.org/10.1007/978-3-540-24618-3_18}, \href
  {http://dx.doi.org/10.1007/978-3-540-24618-3_18}
  {\path{doi:10.1007/978-3-540-24618-3_18}}.

\bibitem{DBLP:conf/compgeom/FulekK18}
Radoslav Fulek and Jan Kyncl.
\newblock Hanani-tutte for approximating maps of graphs.
\newblock In Bettina Speckmann and Csaba~D. T{\'{o}}th, editors, {\em
  {S}o{CG}~'18}, volume~99 of {\em LIPIcs}, pages 39:1--39:15. Schloss Dagstuhl
  - Leibniz-Zentrum fuer Informatik, 2018.
\newblock URL: \url{https://doi.org/10.4230/LIPIcs.SoCG.2018.39}, \href
  {http://dx.doi.org/10.4230/LIPIcs.SoCG.2018.39}
  {\path{doi:10.4230/LIPIcs.SoCG.2018.39}}.

\bibitem{DBLP:journals/corr/abs-1305-4519}
Radoslav Fulek, Jan Kyncl, Igor Malinovic, and D{\"{o}}m{\"{o}}t{\"{o}}r
  P{\'{a}}lv{\"{o}}lgyi.
\newblock Efficient c-planarity testing algebraically.
\newblock {\em CoRR}, abs/1305.4519, 2013.
\newblock URL: \url{http://arxiv.org/abs/1305.4519}, \href
  {http://arxiv.org/abs/1305.4519} {\path{arXiv:1305.4519}}.

\bibitem{FulekKMP15}
Radoslav Fulek, Jan Kyncl, Igor Malinovic, and D{\"{o}}m{\"{o}}t{\"{o}}r
  P{\'{a}}lv{\"{o}}lgyi.
\newblock Clustered planarity testing revisited.
\newblock {\em Electr. J. Comb.}, 22(4):P4.24, 2015.
\newblock URL:
  \url{http://www.combinatorics.org/ojs/index.php/eljc/article/view/v22i4p24}.

\bibitem{FulekTothSODA19}
Radoslav Fulek and Csaba~D. T{\'{o}}th.
\newblock Atomic embeddability, clustered planarity, and thickenability.
\newblock {\em CoRR}, abs/1907.13086, 2019.
\newblock URL: \url{http://arxiv.org/abs/1907.13086}, \href
  {http://arxiv.org/abs/1907.13086} {\path{arXiv:1907.13086}}.

\bibitem{DBLP:books/daglib/0037866}
Christopher~D. Godsil and Gordon~F. Royle.
\newblock {\em Algebraic Graph Theory}.
\newblock Graduate texts in mathematics. Springer, 2001.
\newblock URL: \url{https://doi.org/10.1007/978-1-4613-0163-9}, \href
  {http://dx.doi.org/10.1007/978-1-4613-0163-9}
  {\path{doi:10.1007/978-1-4613-0163-9}}.

\bibitem{Goodrich2006}
Michael~T. Goodrich, George~S. Lueker, and Jonathan~Z. Sun.
\newblock C-planarity of extrovert clustered graphs.
\newblock In Patrick Healy and Nikola~S. Nikolov, editors, {\em {GD} '05},
  volume 3843 of {\em LNCS}, pages 211--222. Springer, 2005.
\newblock URL: \url{https://doi.org/10.1007/11618058_20}, \href
  {http://dx.doi.org/10.1007/11618058_20} {\path{doi:10.1007/11618058_20}}.

\bibitem{DBLP:journals/talg/GuT08}
Qian{-}Ping Gu and Hisao Tamaki.
\newblock Optimal branch-decomposition of planar graphs in {O}$(n^3)$ time.
\newblock {\em {ACM} Trans. Algorithms}, 4(3):30:1--30:13, 2008.
\newblock URL: \url{http://doi.acm.org/10.1145/1367064.1367070}, \href
  {http://dx.doi.org/10.1145/1367064.1367070}
  {\path{doi:10.1145/1367064.1367070}}.

\bibitem{Gutwenger2002}
Carsten Gutwenger, Michael J{\"{u}}nger, Sebastian Leipert, Petra Mutzel,
  Merijam Percan, and Ren{\'{e}} Weiskircher.
\newblock Advances in c-planarity testing of clustered graphs.
\newblock In Stephen~G. Kobourov and Michael~T. Goodrich, editors, {\em {GD}
  '02}, volume 2528 of {\em LNCS}, pages 220--235. Springer, 2002.
\newblock URL: \url{https://doi.org/10.1007/3-540-36151-0_21}, \href
  {http://dx.doi.org/10.1007/3-540-36151-0_21}
  {\path{doi:10.1007/3-540-36151-0_21}}.

\bibitem{DBLP:conf/alenex/GutwengerMS14}
Carsten Gutwenger, Petra Mutzel, and Marcus Schaefer.
\newblock Practical experience with hanani-tutte for testing c-planarity.
\newblock In Catherine~C. McGeoch and Ulrich Meyer, editors, {\em
  {ALENEX}~'14}, pages 86--97. {SIAM}, 2014.
\newblock URL: \url{https://doi.org/10.1137/1.9781611973198.9}, \href
  {http://dx.doi.org/10.1137/1.9781611973198.9}
  {\path{doi:10.1137/1.9781611973198.9}}.

\bibitem{DBLP:journals/jda/HongN10}
Seok{-}Hee Hong and Hiroshi Nagamochi.
\newblock Convex drawings of hierarchical planar graphs and clustered planar
  graphs.
\newblock {\em J. Discrete Algorithms}, 8(3):282--295, 2010.
\newblock URL: \url{https://doi.org/10.1016/j.jda.2009.05.003}, \href
  {http://dx.doi.org/10.1016/j.jda.2009.05.003}
  {\path{doi:10.1016/j.jda.2009.05.003}}.

\bibitem{hn-sat2pepg-14}
Seok-Hee Hong and Hiroshi Nagamochi.
\newblock Simpler algorithms for testing two-page book embedding of partitioned
  graphs.
\newblock {\em Theoretical Computer Science}, 2016.
\newblock URL:
  \url{http://www.sciencedirect.com/science/article/pii/S0304397515012207},
  \href {http://dx.doi.org/https://doi.org/10.1016/j.tcs.2015.12.039}
  {\path{doi:https://doi.org/10.1016/j.tcs.2015.12.039}}.

\bibitem{DBLP:journals/cacm/HopcroftT73}
John~E. Hopcroft and Robert~Endre Tarjan.
\newblock Efficient algorithms for graph manipulation {[H]} (algorithm 447).
\newblock {\em Commun. {ACM}}, 16(6):372--378, 1973.
\newblock URL: \url{http://doi.acm.org/10.1145/362248.362272}, \href
  {http://dx.doi.org/10.1145/362248.362272} {\path{doi:10.1145/362248.362272}}.

\bibitem{JelinekJKL08}
V{\'{\i}}t Jel{\'{\i}}nek, Eva Jel{\'{\i}}nkov{\'{a}}, Jan Kratochv{\'{\i}}l,
  and Bernard Lidick{\'{y}}.
\newblock Clustered planarity: Embedded clustered graphs with two-component
  clusters.
\newblock In Ioannis~G. Tollis and Maurizio Patrignani, editors, {\em {GD}
  '08}, volume 5417 of {\em LNCS}, pages 121--132. Springer, 2008.
\newblock URL: \url{https://doi.org/10.1007/978-3-642-00219-9_13}, \href
  {http://dx.doi.org/10.1007/978-3-642-00219-9_13}
  {\path{doi:10.1007/978-3-642-00219-9_13}}.

\bibitem{Jelinkova2008}
Eva Jel{\'{\i}}nkov{\'{a}}, Jan K{\'{a}}ra, Jan Kratochv{\'{\i}}l, Martin
  Pergel, Ondrej Such{\'{y}}, and Tom{\'{a}}s Vyskocil.
\newblock Clustered planarity: Small clusters in cycles and eulerian graphs.
\newblock {\em J. Graph Algorithms Appl.}, 13(3):379--422, 2009.
\newblock URL: \url{http://jgaa.info/accepted/2009/Jelinkova+2009.13.3.pdf}.

\bibitem{DBLP:journals/dam/NagamochiK07}
Hiroshi Nagamochi and Katsutoshi Kuroya.
\newblock Drawing c-planar biconnected clustered graphs.
\newblock {\em Discrete Applied Mathematics}, 155(9):1155--1174, 2007.
\newblock URL: \url{https://doi.org/10.1016/j.dam.2006.04.044}, \href
  {http://dx.doi.org/10.1016/j.dam.2006.04.044}
  {\path{doi:10.1016/j.dam.2006.04.044}}.

\bibitem{DBLP:journals/jgt/Oum08}
Sang{-}il Oum.
\newblock Rank-width is less than or equal to branch-width.
\newblock {\em Journal of Graph Theory}, 57(3):239--244, 2008.
\newblock URL: \url{https://doi.org/10.1002/jgt.20280}, \href
  {http://dx.doi.org/10.1002/jgt.20280} {\path{doi:10.1002/jgt.20280}}.

\bibitem{DBLP:journals/jct/RobertsonS91}
Neil Robertson and Paul~D. Seymour.
\newblock Graph minors. {X}. {Obstructions} to tree-decomposition.
\newblock {\em Journal of Combinatorial Theory, Series {B}}, 52(2):153--190,
  1991.
\newblock URL: \url{https://doi.org/10.1016/0095-8956(91)90061-N}, \href
  {http://dx.doi.org/10.1016/0095-8956(91)90061-N}
  {\path{doi:10.1016/0095-8956(91)90061-N}}.

\bibitem{DBLP:journals/talg/RueST14}
Juanjo Ru{\'{e}}, Ignasi Sau, and Dimitrios~M. Thilikos.
\newblock Dynamic programming for graphs on surfaces.
\newblock {\em {ACM} Trans. Algorithms}, 10(2):8:1--8:26, 2014.
\newblock URL: \url{http://doi.acm.org/10.1145/2556952}, \href
  {http://dx.doi.org/10.1145/2556952} {\path{doi:10.1145/2556952}}.

\bibitem{cgp}
R\'obert Sas\'ak.
\newblock {\em Comparing 17 graph parameters}.
\newblock Master's thesis, Department of Informatics, University of Bergen,
  Bergen, Norway, 2010.

\bibitem{DBLP:journals/jgaa/Schaefer13}
Marcus Schaefer.
\newblock Toward a theory of planarity: Hanani-tutte and planarity variants.
\newblock {\em J. Graph Algorithms Appl.}, 17(4):367--440, 2013.
\newblock URL: \url{https://doi.org/10.7155/jgaa.00298}, \href
  {http://dx.doi.org/10.7155/jgaa.00298} {\path{doi:10.7155/jgaa.00298}}.

\bibitem{DBLP:journals/combinatorica/SeymourT94}
Paul~D. Seymour and Robin Thomas.
\newblock Call routing and the ratcatcher.
\newblock {\em Combinatorica}, 14(2):217--241, 1994.
\newblock URL: \url{https://doi.org/10.1007/BF01215352}, \href
  {http://dx.doi.org/10.1007/BF01215352} {\path{doi:10.1007/BF01215352}}.

\bibitem{stanley_2015}
Richard~P. Stanley.
\newblock {\em Catalan Numbers}.
\newblock Cambridge University Press, 2015.
\newblock \href {http://dx.doi.org/10.1017/CBO9781139871495}
  {\path{doi:10.1017/CBO9781139871495}}.

\bibitem{DBLP:conf/isaac/ThilikosSB00}
Dimitrios~M. Thilikos, Maria~J. Serna, and Hans~L. Bodlaender.
\newblock Constructive linear time algorithms for small cutwidth and
  carving-width.
\newblock In D.~T. Lee and Shang{-}Hua Teng, editors, {\em {ISAAC}~'00}, volume
  1969 of {\em LNCS}, pages 192--203. Springer, 2000.
\newblock URL: \url{https://doi.org/10.1007/3-540-40996-3_17}, \href
  {http://dx.doi.org/10.1007/3-540-40996-3_17}
  {\path{doi:10.1007/3-540-40996-3_17}}.

\bibitem{DBLP:journals/corr/abs-1907-01630}
Juan Jos{\'{e}}~Besa Vial, Giordano {Da Lozzo}, and Michael~T. Goodrich.
\newblock Computing k-modal embeddings of planar digraphs.
\newblock In Ola~Svensson Michael A.~Bender and Grzegorz Herman, editors, {\em
  {ESA} 2019}, volume 144 of {\em LIPIcs}, pages 17:1--17:16. Schloss Dagstuhl
  - Leibniz-Zentrum fuer Informatik, 2019.
\newblock \href {http://dx.doi.org/10.4230/LIPIcs.ESA.2019.17}
  {\path{doi:10.4230/LIPIcs.ESA.2019.17}}.

\end{thebibliography}

\end{document}